\pdfoutput=1
\documentclass[a4paper,USenglish]{lipics-v2021}
\hideLIPIcs \nolinenumbers
\usepackage{amstext,amsgen,latexsym,amsmath}
\usepackage{amssymb,amsfonts}
\usepackage{amsthm} 

\usepackage{eclbkbox}

\usepackage{framed}
\usepackage{tcolorbox}
\tcbuselibrary{breakable}

\usepackage{pifont}
\usepackage{hyperref}
\usepackage{mdframed}
\usepackage{url}

\newtheorem{fact}{Fact}
\newtheorem{dfn}{Definition}
\newtheorem{cor}{Corollary}
\newtheorem{lem}{Lemma}

\def\tr{\mathrm{tr}}

\newcommand{\ket}[1]{\lvert #1 \rangle}
\newcommand{\bra}[1]{\langle #1 \lvert}
\newcommand{\SGDI}{\mathsf{SGDI}_r}
\newcommand{\SE}{\mathsf{SetEquality}_{\ell,U}}
\newcommand{\EQ}{\mathsf{EQ}_n^t}
\newcommand{\ignore}[1]{}

\title{Distributed Merlin-Arthur Synthesis of Quantum States and Its Applications} %TODO Please add

\author{Fran\c{c}ois Le Gall}{Graduate School of Mathematics, Nagoya University, Nagoya, Japan}{legall@math.nagoya-u.ac.jp}{}{}
\author{Masayuki Miyamoto}{Graduate School of Mathematics, Nagoya University, Nagoya, Japan }{masayuki.miyamoto95@gmail.com}{}{}
\author{Harumichi Nishimura}{Graduate School of Informatics, Nagoya University, Nagoya, Japan} {hnishimura@is.nagoya-u.ac.jp}{}{}%TODO mandatory, please use full name; only 1 author per \author macro; first two parameters are mandatory, other parameters can be empty. Please provide at least the name of the affiliation and the country. The full address is optional. Use additional curly braces to indicate the correct name splitting when the last name consists of multiple name parts.

\authorrunning{F. Le Gall, M. Miyamoto and H. Nishimura} %TODO mandatory. First: Use abbreviated first/middle names. Second (only in severe cases): Use first author plus 'et al.'

\Copyright{Fran\c{c}ois Le Gall, Masayuki Miyamoto and Harumichi Nishimura} %TODO mandatory, please use full first names. LIPIcs license is "CC-BY";  http://creativecommons.org/licenses/by/3.0/

 %TODO mandatory: Please choose ACM 2012 classifications from https://dl.acm.org/ccs/ccs_flat.cfm 

 %TODO mandatory; please add comma-separated list of keywords

\category{} %optional, e.g. invited paper

\EventEditors{John Q. Open and Joan R. Access}
\EventNoEds{2}
\EventLongTitle{42nd Conference on Very Important Topics (CVIT 2016)}
\EventShortTitle{CVIT 2016}
\EventAcronym{CVIT}
\EventYear{2016}
\EventDate{December 24--27, 2016}
\EventLocation{Little Whinging, United Kingdom}
\EventLogo{}
\SeriesVolume{42}
\ArticleNo{23}
\begin{document}
\maketitle
\begin{abstract}
The generation and verification of quantum states are fundamental tasks for quantum information processing that have recently been investigated by Irani, Natarajan, Nirkhe, Rao and Yuen [CCC 2022], Rosenthal and Yuen [ITCS 2022], Metger and Yuen [FOCS 2023] under the term \emph{state synthesis}. This paper studies this concept from the viewpoint of quantum distributed computing, and especially distributed quantum Merlin-Arthur (dQMA) protocols. We first introduce a novel task, on a line, called state generation with distributed inputs (SGDI). In this task, the goal is to generate the quantum state $U\ket{\psi}$ at the rightmost node of the line, where $\ket{\psi}$ is a quantum state given at the leftmost node and $U$ is a unitary matrix whose description is distributed over the nodes of the line. We give a dQMA protocol for SGDI and utilize this protocol to construct a dQMA protocol for the Set Equality problem studied by Naor, Parter and Yogev [SODA 2020], and complement our protocol by showing classical lower bounds for this problem. Our second contribution is a dQMA protocol, based on a recent work by Zhu and Hayashi [Physical Review A, 2019], to create EPR-pairs between adjacent nodes of a network without quantum communication. As an application of this dQMA protocol, we prove a general result showing how to convert any dQMA protocol on an arbitrary network into another dQMA protocol where the verification stage does not require any quantum communication.

\end{abstract}

\maketitle
\section{Introduction}\label{sec:introduction}
%\paragraph*{Background.}
While quantum computational complexity has so far mostly investigated the complexity of classical problems (e.g., computing Boolean functions) in the quantum setting, recent works \cite{Aaronson16, INNRY22, JLS18, Kretschmer21, 2301.07730, RY22} have started investigating the complexity of \emph{quantum} problems (e.g., generating quantum states). For instance, Ji, Liu and Song~\cite{JLS18} and Kretschmer \cite{Kretschmer21} have investigated the concept of quantum pseudorandom states from complexity-theoretic and cryptographic perspectives. Irani, Natarajan, Nirkhe, Rao, and Yuen \cite{INNRY22} have made in-depth investigations of the complexity of the \emph{state synthesis problem} in a setting first introduced by Aaronson \cite{Aaronson16} where the goal is to generate a quantum state by making queries to a classical oracle encoding the state.  Rosenthal and Yuen~\cite{RY22} and Metger and Yuen~\cite{2301.07730} have considered interactive proofs for synthesizing quantum states (and also for implementing unitaries). Here the main goal is to generate complicated quantum states (e.g., quantum states described by an exponential-size generating quantum circuit) efficiently with the help of an all-powerful but untrusted prover. Note that in settings where an all-powerful prover is present, the task of quantum state synthesis is closely related to the task of quantum state verification (since the prover can simply send the quantum state that needs to be synthesized).  
%and considered the class of quantum states that can be efficiently computed by interactive protocols.

In this paper, we investigate the task of state generation and verification in the setting of quantum distributed computing. Quantum distributed computing is a fairly recent research topic: despite early investigations in the 2000s and the 2010s \cite{Arfaoui+14,Denchev+08,ElkinKNP14,GavoilleKM09,Tani+12}, it is only in the past five years that significant advances have been done in understanding the power of quantum distributed algorithms \cite{AVPODC22,FLNP21,Izumi+PODC19,Izumi+STACS20,GallM18,LeGall+STACS19,WYPODC22}. Fraigniaud, Le Gall, Nishimura, and Paz~\cite{FLNP21}, in particular, have investigated the power of distributed quantum proofs in distributed computing, which is the natural quantum version of the concept of distributed classical proofs (also called locally-checkable proofs~\cite{GoosS16} or proof-labeling schemes \cite{KormanKP10}): each node of the network receives, additionally to its input, a quantum state (called a quantum proof) from an all-powerful but untrusted party called the prover. The main result from \cite{FLNP21} shows that there exist classical problems that can be solved by quantum protocols using quantum proofs of length exponentially smaller than in the classical case.  

We present two main results about state generation and verification in the setting where an all-powerful but untrusted prover helps the nodes in a non-interactive way, and apply these results to design new quantum protocols for concrete problems studied recently in \cite{FLNP21,NPY20}. 

%===========
\subsection{First result and applications: State Generation with Distributed Inputs}\label{sub:result1}
One of the main conceptual contributions of this paper is introducing the following problem: In a network of $r+1$ nodes $v_0,v_1,\ldots,v_r$, node $v_0$ is given as input an $n$-qubit quantum state $\ket{\psi}$. The goal is to generate the quantum state $U\ket{\psi}$ at node $v_r$, where $U$ is a unitary matrix whose description is distributed over the nodes of the network.
%: $U$ is given as a sequence of unitary matrices $U=U_{r}\cdots U_1$, where the description of $U_j$ is input to node $v_j$, for $j\in\{1,\ldots,r\}$. 
For concreteness, in this paper we focus on the case where the network is a path of length $r$ and the nodes $v_0,v_r$ are both extremities of the path.\footnote{In distributed computing it is standard to first investigate the complexity of computational problems on simple network topologies such as a path or a ring. A solution on the path can often be extended to networks of more complex topology, or be used as a building block for solving problems on network of arbitrary topology.} 
%A solution on a path can then typically be extended to more complex topology. }

Here is the precise description of the problem. The parties $v_0,v_1,\ldots,v_r$ are the nodes of a line graph of length $r$: the left-end extremity is $v_0$, the right-end extremity is $v_r$, and nodes $v_j$ and $v_{j+1}$ are connected for $j=0,1,\ldots,r-1$. Node $v_0$ receives as input the classical description of an $n$-qubit state $|\psi\rangle$, as a $2^n$-dimensional vector.\footnote{Our protocol actually only requires $v_0$ to be able to generate many copies of $\ket{\psi}$, and thus also works when the input is a description of a quantum circuit generating $\ket{\psi}$, or even a black box generating $\ket{\psi}$.} The other nodes $v_j$ for $j=1,2,\ldots,r$ receive as input the description of an $n$-qubit unitary transformation: each node $v_j$ receives the description of a unitary transformation $U_j$ acting on $n$ qubits. In this setting, the aim is to generate the quantum state
\[
|\varphi_r\rangle:=U_{r}\cdots U_1|\psi\rangle
\] 
at the right-end extremity $v_r$. 
We call this problem $n$-qubit 
{\em State Generation with Distributed Inputs} on the line of length~$r$ ($n$-qubit $\SGDI$). 
Without a prover, this problem is clearly not solvable in less than $r$ rounds of communications between neighbors (this can be seen easily by considering the case where $U_1=\cdots=U_r=I$). 

We consider the setting where a prover (an all-powerful but untrusted party) helps the nodes in a non-interactive way: at the very beginning of the protocol the prover sends to node $v_j$ a quantum state $\rho_j$ of at most~$s_c$ qubits, for each $j\in\{0,1,\ldots,r\}$. Here $s_c$ is called the certificate size of the protocol and the state $\rho_j$ is called the certificate to $v_j$.
%The global quantum state sent by the prover (which consists of at most $(r+1)s_c$ qubits and such that its reduced states are $\rho_0,\ldots,\rho_r$) is called the certificate of the protocol.
The nodes then run a one-round\footnote{As in almost all prior works on (classical or quantum) distributed proofs, in this paper we consider only one-round verification algorithms.} distributed quantum algorithm (called the verification algorithm). More precisely, the nodes first perform one round of (synchronous) communication: each node sends one quantum message of at most $s_m$ qubits to its neighbors ($s_m$ is called the message size of the protocol). Each node then decides to either accept or reject. Such protocols, which have been introduced and studied in \cite{FLNP21}, are called distributed Quantum Merlin-Arthur (dQMA) protocols (see Section \ref{sub:dQMA} for details). Additionally, when considering dQMA protocols for $n$-qubit $\SGDI$, we add the requirement that node $v_r$ outputs an $n$-qubit quantum state at the end of the protocol. 

Here is our main result: 

\begin{theorem}\label{thm:state-transfer-simplified} 
For any constant $\varepsilon>0$, there exists a dQMA protocol for $n$-qubit $\SGDI$ with certificate size $O(n^2r^5)$ and message size $O(nr^2)$ 
satisfying the following: 
({\bf completeness}) 
%There is a certificate such that 
There are certificates $\rho_0,\ldots,\rho_r$ such that
all the nodes accept and node $v_r$ outputs 
$|\varphi_r\rangle$ with probability $1$; 
({\bf soundness}) If all the nodes accept with probability 
at least $\varepsilon$, 
then the output state $\rho$ of node $v_r$ satisfies 
$\langle\varphi_r|\rho|\varphi_r \rangle\geq 1-\varepsilon$.
\end{theorem}
The protocol of Theorem \ref{thm:state-transfer-simplified} is a dQMA protocol with perfect completeness and soundness $\varepsilon$. Indeed, when receiving appropriate certificates from the prover, all nodes accept with probability 1 and node $v_r$ outputs the state $\ket{\varphi_r}$. On the other hand, if the state $\rho$ is far from $\ket{\varphi_r}$, the soundness condition guarantees that for any certificates $\rho_0,\ldots,\rho_r$ received from the prover (including the case of entangled certificates), the probability that at least one node rejects is at least $1-\varepsilon$ (remember that the quantity $\langle\varphi_r|\rho|\varphi_r \rangle$ represents the square root of the fidelity between $\ket{\varphi_r}\bra{\varphi_r}$ and $\rho$ --- see Section~\ref{sub:prelim} for details).

% \begin{theorem}\label{thm:state-transfer-simplified} 
% For any constant $\varepsilon>0$, 
% there exists a dQMA protocol 
% for $n$-qubit SGDI with on the line of length $r$ 
% with certificate size $O(n^2r^5)$ and message size $O(nr^2)$ 
% satisfying the following: 
% ({\bf completeness}) There is a certificate such that 
% all the nodes accept and node $v_r$ outputs 
% $|\varphi_r\rangle$ with probability $1$; 
% ({\bf soundness}) If all the nodes accept with probability 
% at least $\varepsilon$, 
% then the output state $\rho$ of node $v_r$ satisfies 
% $\langle\varphi_r|\rho|\varphi_r \rangle\geq 1-\varepsilon$.
% \end{theorem}

As an application of Theorem \ref{thm:state-transfer-simplified}, we construct a quantum protocol for a concrete computational task called Set Equality, which was introduced in Ref.~\cite{NPY20}. Here is the formal definition over a network of arbitrary topology (represented by an arbitrary graph $G=(V,E)$).

\begin{dfn}[$\SE$~\cite{NPY20}]\label{def:seteq}
Let $\ell$ be a positive integer and $U$ be a finite set.
Each node $u$ of a graph $G=(V,E)$ holds two lists of $\ell$ elements $(a_{u,1},\ldots,a_{u,\ell})$ 
and $(b_{u,1},\ldots,b_{u,\ell})$ as input, where $a_{u,i},b_{u,i}\in U$ for all $i\in\{1,2,\ldots,\ell\}$.
Define $A=\{a_{u,i}\mid u\in V,\ i\in \{1,2,\ldots,\ell\} \}$ and  $B=\{b_{u,i}\mid u\in V,\ i\in \{1,2,\ldots,\ell\} \}$. 
The output of $\SE$ is $1$ (yes), 
if $A=B$ as multisets and $0$ (no) otherwise.
\end{dfn}

Using Theorem \ref{thm:state-transfer-simplified} we obtain the following result:
\begin{theorem}\label{thm:SetEQ}
For any small enough constant $\varepsilon>0$, %and any parameter $l=o(\frac{2^N}{r})$, 
there exists a dQMA protocol for $\SE$ on the line graph of length $r$ 
with completeness $1-\varepsilon$ and 
soundness $\varepsilon$ that has
certificate size $O(r^{5}\log^2(\ell r)\log^2 |U|)$ 
and message size $O(r^{2}\log(\ell r)\log |U|)$. 
\end{theorem}
While Ref.~\cite{NPY20} considered the special case of $\SE$ and showed efficient distributed \textit{interactive} protocols with small certificate and message size (see Section~\ref{subsec:related work}), no (nontrivial) classical dMA protocol (or lower bound) is known before this paper to our best knowledge. %It is easy to see that a trivial dMA protocol has certificate and message size $O(r\ell \log|U|)$ (the prover sends all of the inputs to all nodes). 
We complement the result in Theorem~\ref{thm:SetEQ} by showing classical lower bounds and upper bounds of distributed Merlin-Arthur (dMA) protocols for $\SE$.

\begin{theorem}\label{thm:SetEQ_classical_lower_bound}
For any dMA protocol for $\SE$ on a line graph of length $r$ 
with certificate size $s_c$, completeness $3/4$, and soundness $1/4$,
\begin{itemize}
\item if $|U| < \ell$, then $s_c = \Omega(|U|\log (\ell/|U|))$;
%\item If $|U| \geq r$, then $m = \Omega(r\log (\ell))$;
\item if $|U|= \Omega(\ell)$, then $s_c = \Omega(\ell)$;
\item if $|U| =\Omega(r\ell)$, then $s_c = \Omega(r\ell)$.
\end{itemize}
\end{theorem}

\begin{theorem}\label{thm:SetEQ_classical_upper_bound}
There exists a dMA protocol for $\SE$ on a line graph of length $r$ 
with completeness $1$ and 
soundness $0$ whose
certificate size and message size are both $O(\mathrm{min}\{r\ell \log |U|,|U| \log (r\ell)\})$.
\end{theorem}
Although the dependence in $r$ is worse than in the classical dMA protocol of Theorem~\ref{thm:SetEQ_classical_upper_bound}, the dependence of the dQMA protocol of Theorem~\ref{thm:SetEQ} in $\ell$ (the number of elements each node receives) and $|U|$ (the size of the universal set) are polylogarithmic. On the other hand, in classical case, we have linear lower bounds with respect to $\ell$ and $|U|$ as in Theorem~\ref{thm:SetEQ_classical_lower_bound}. Therefore Theorem~\ref{thm:SetEQ} gives a significant improvement for sufficiently large $\ell$ and $|U|$. 
%Although the dependence in $r$ and $|U|$ is worse than in the trivial dMA protocol, the dependence of the dQMA protocol of Theorem \ref{thm:SetEQ} in~$\ell $ (the number of elements each node receives) is exponentially better and thus gives a significant improvement for sufficiently large $\ell $. 
This assumption about the input parameters seems reasonable when considering applications similar to those of the dQMA protocol for the equality problem proposed in Ref.~\cite{FLNP21}. 
Note that our bounds of classical certificate size in Theorem~\ref{thm:SetEQ_classical_lower_bound} and Theorem~\ref{thm:SetEQ_classical_upper_bound} are tight up to $\mathrm{poly}\log(\ell, |U|, r)$ factors when $|U| <\ell$ or $|U| = \Omega(r\ell)$.

%=========================
\subsection{Second result and applications: EPR-pairs generation and LOCC dQMA protocols}
Our second contribution is a protocol, based on a recent work by Zhu and Hayashi \cite{ZH19PRA}, to create EPR-pairs between adjacent nodes of a network without quantum communication in the same setting as above, where a prover helps the nodes in a non-interactive way. As an application of this protocol, we prove a general result showing how to convert any dQMA protocol on an arbitrary network into another dQMA protocol where the verification algorithm uses only classical communication (instead as quantum communication, as allowed in the definition of dQMA protocols and used in all dQMA protocols of Ref.~\cite{FLNP21} and Theorems \ref{thm:state-transfer-simplified} and \ref{thm:SetEQ} above). 

More precisely, we say a dQMA protocol is an LOCC (Local Operation and Classical Communication) dQMA protocol  if the verification algorithm can be implemented only by local operations at each node and classical communication between neighboring nodes (i.e., no quantum communication is allowed). Our protocol for generating EPR-pairs enables us to show the following theorem: 

\begin{theorem}\label{thm:convertion-LOCC}
For any constant $p_c$ and $p_s$ such that $0\leq p_s<p_c\leq 1$, 
let ${\cal P}$ be a dQMA protocol for some problem on a network $G$ 
with completeness $p_c$, soundness $p_s$,
certificate size $s_c^{{\cal P}}$ and message size $s_m^{{\cal P}}$.
For any small enough constant $\gamma>0$,
there exists an LOCC dQMA protocol ${\cal P}'$ for the same problem on $G$
with completeness $p_c$, soundness $p_s+\gamma$, 
certificate size $s_c^{{\cal P}}+O(d_{\max} s_m^{{\cal P}} s_{tm}^{{\cal P}})$, 
and message size $O(s_m^{{\cal P}} s_{tm}^{{\cal P}})$, 
where
$d_{\max}$ is the maximum degree of $G$, and 
$s_{tm}^{{\cal P}}$ is the total number of qubits sent in the verification stage of ${\cal P}$.    
\end{theorem}

As an application of Theorem \ref{thm:convertion-LOCC}, we consider the equality problem studied in Ref.~\cite{FLNP21}. In this problem, denoted $\EQ$, a collection of $n$-bit strings $x_1,x_2,\ldots, x_t$ is given as input to $t$ specific nodes $u_1,u_2,\ldots, u_t$ (called terminals) of an arbitrary network $G = (V,E)$ as follows: node $u_i$ receives $x_i$, for $i\in\{1,2,\ldots,t\}$. The goal is to check whether the~$t$ strings are equal, i.e., whether $x_1 = \cdots = x_t$. By applying Theorem \ref{thm:convertion-LOCC} to the main result in Ref.~\cite{FLNP21} (a dQMA protocol for $\EQ$ with certificate size $O(tr^2\log n)$ and message size $O(tr^2\log(n+r))$), we obtain the following corollary: 
 
\begin{cor}
For any small enough constant $\varepsilon>0$,
there is an LOCC dQMA protocol for $\EQ$ 
with completeness~$1$, soundness $\varepsilon$,
certificate size $O(d_{\max}|V|t^2r^4 \log^2 (n+r))$ 
and messages size $O(|V| t^2 r^4 \log^2(n + r))$, 
where $r$ is the radius of the set of the $t$ terminals and $|V|$ is the number of nodes of the network $G=(V,E)$. 
\end{cor}

We can also apply Theorem~\ref{thm:convertion-LOCC} to the dQMA protocol of Theorem~\ref{thm:SetEQ}, leading to the following corollary:

\begin{cor}
For any small enough constant $\varepsilon>0$,
there is an LOCC dQMA protocol for $\SE$ on the line
graph of length $r$ with completeness~$1-\varepsilon$, soundness $\varepsilon$,
certificate size $O(r^5 \log^2 (\ell r)\log^2|U|)$ 
and messages size $O(r^5 \log^2 (\ell r)\log^2|U|)$.
\end{cor}

Note that these LOCC dQMA protocols still have good dependence in the main parameters we are interested in: the parameter $n$ for $\EQ$ (for which the dependence is still exponentially better than any classical dMA protocols) 
%(the classical lower bound is shown in \cite{FLNP21}), 
and the parameters $\ell$ and $|U|$ for $\SE$ (for which the dependence is still exponentially better than any classical dMA protocols, due to Theorem~\ref{thm:SetEQ_classical_lower_bound}).

\subsection{Overview of our proofs}

To explain the proof idea of Theorem~\ref{thm:state-transfer-simplified}, 
we only consider the simplified case $U_1=\cdots=U_r=I$. 
The general case can be proved similarly by a slightly more complicated analysis.
 
The dQMA protocol to prove Theorem~\ref{thm:state-transfer-simplified} 
is based on the dQMA protocol on the line of length $r$ by Fraigniaud et al.~\cite{FLNP21}. 
In the setting of Ref.~\cite{FLNP21}, the left-end extremity $v_0$ has an $n$-bit string $x$, the right-end extremity $v_r$ has an $n$-bit string $y$, and the other intermediate nodes have no input. The goal is to verify whether $x=y$. 
The dQMA protocol in Ref.~\cite{FLNP21} checks whether the fingerprint state $|\psi_0\rangle=|\psi_x\rangle$~\cite{BCWW01} prepared by $v_0$ is equal to the fingerprint state $|\psi_{r}\rangle=|\psi_y\rangle$ 
prepared by $v_r$ ($x=y$), 
or $|\psi_0\rangle$ is almost orthogonal to $|\psi_r\rangle$ ($x\neq y$). 
For this, node $v_j$ ($2\leq j\leq r-1$) receives a subsystem whose reduced state is $\rho_j$ as a certificate from the prover. 
At the verification stage, any node (except for $v_r$) chooses keeping its certificate by itself, 
or sending it to the right neighboring node with probability $1/2$ to check if the reduced states of the two neighboring nodes, $\rho_j$ and $\rho_{j+1}$, 
are close, which can be checked by the SWAP test~\cite{BCWW01} (using Lemma~\ref{lem:FLNP21}). 
If $x=y$, then the prover can send $|\psi_0\rangle~(=|\psi_r\rangle)$ for every intermediate node to pass all the SWAP tests done at the verification stage, 
which means accept. Otherwise, the SWAP test done at some node rejects with a reasonable probability 
since $|\psi_x\rangle$ is very far from $|\psi_y\rangle$, and hence the distance between $\rho_{j}$ and $\rho_{j+1}$ should be far at some $j$. 

Now the case that $U_1=\cdots=U_r=I$ (which means that all nodes except $v_0$ have no input) in the setting of $\SGDI$  (then the goal state $|\varphi_r\rangle$ at $v_r$ is the same as the state $|\psi\rangle$ of $v_0$) is similar to the setting of Ref.~\cite{FLNP21}, except that $v_r$ also has no input. 
 The difficulty is that $v_r$ has no state that can be generated by itself, and thus the analysis of Ref.~\cite{FLNP21} cannot be used as it is. 

To overcome this difficulty, we utilize an idea from the verification of graph states~\cite{HM15PRL,MTH17PRA}, 
in particular, the idea by Morimae, Takeuchi, and Hayashi~\cite{MTH17PRA}. 
They used the following basic idea for their protocol in order to verify an arbitrary graph state $|G\rangle$ sent from the prover (or prepared by a malicious party): 
(i) the verifier receives $(m+k+1)$ subsystems, in which each subsystem ideally contains $|G\rangle$, from the prover; 
(ii) the verifier chooses $m$ subsystems uniformly at random, and discards them; 
(iii) the verifier chooses one subsystem, and some test that $|G\rangle$ should pass (stabilizer test) is done for each of the remaining $k$ subsystems; and 
(iv) if all the tests passed, the chosen subsystem in (iii) should be close to $|G\rangle$, which is proved by using a quantum de Finetti theorem with some measurement condition~\cite{LS15PRL}  (exponentially better in the dimension of the subsystem than the standard quantum de Finetti theorem~\cite{CKMR07CMP}). 
Note that (ii) and (iii) are necessary since the assumption that the total system is permutation-invariant is needed to apply the quantum de Finetti theorem. 
  
Our protocol applies the idea of Ref.~\cite{MTH17PRA} to the verification protocol of Ref.~\cite{FLNP21} explained above. 
Namely, the parties $v_1,v_2,\ldots,v_r$ first receives $(m+k+1)$ subsystems, where each subsystem ideally contains $|\psi\rangle^{\otimes r}$, sent from the prover.
For $k$ subsystems that are randomly chosen, 
we apply the verification protocol of Ref.~\cite{FLNP21}.
Actually, we have a subtle problem with the corresponding steps of (ii) and (iii) in the idea of Ref.~\cite{MTH17PRA}, 
since $v_0,v_1,\ldots,v_r$ do not have any shared randomness, 
and thus those steps cannot be implemented jointly. 
Fortunately, this problem can be overcome 
since the permutation-invariant property is satisfied by the random permutations of $(m+k+1)$ subsystems on {\em each} party. 

The dQMA protocol for Theorem~\ref{thm:SetEQ} is based on the distributed interactive protocol by Naor, Parter, and Yogev~\cite{NPY20} using shared randomness\footnote{While there is no shared randomness in their setting, shared randomness can be simulated by two interactions between the prover and the verifier.}. 
In our setting (line of length $r$), the distributed interactive protocol of Ref.~\cite{NPY20} is as follows with two polynomials $\alpha_j(x):=\prod_{i} (x-a_{j,i})$ and $\beta_j(x):=\prod_{i}(x-b_{j,i})$: 
with shared randomness $s$ (taken from a large field), (i) $v_0$ prepares $A_0(s):=\alpha_0(s)$ and $B_0(s):=\beta_0(s)$;
(ii) $v_j$ ($j=1,2,\ldots,r$) ideally receives $A_j(s):=\alpha_0(s)\cdots \alpha_j(s)$ and $B_j(s):=\beta_0(s)\cdots \beta_j(s)$ from the prover;
(iii) $A_j(s)=\alpha_j(s) A_{j-1}(s)$ and $B_j(s)=\beta_j(s) B_{j-1}(s)$ are checked for consistency by communication from $v_{j-1}$ to $v_j$. 
We can see that when $A=B$, $A_r(s)=B_r(s)$ for any $s$, and thus this protocol accepts with probability $1$ 
by the ideal certificates from the prover, while when $A\neq B$,  $A_r(s) \neq B_r(s)$ for most of $s$, and thus some node rejects with reasonable probability.  

Actually, neither interaction nor shared randomness is available in our setting. Instead, we reduce the protocol by Naor et al.~to $\SGDI$ 
with $|\psi\rangle
=|\psi_A\rangle\otimes |\psi_B\rangle$ 
where $|\psi_A\rangle=\sum_s |s\rangle|\alpha_0(s)\rangle$, 
and $|\psi_B\rangle=\sum_s |s\rangle|\beta_0(s)\rangle$, 
and $U=U_{j,A}\otimes U_{j,B}$, 
where 
$U_{j,A}$ roughly\footnote{We actually need some modifications for $U_{j,A}$ to be unitary.} maps $|s\rangle|t\rangle$ to $|s\rangle|\alpha_j(s)t\rangle$ ($j=1,2,\ldots,r$) 
and $U_{j,B}$ roughly maps $|s\rangle|t\rangle$ to $|s\rangle|\beta_j(s)t\rangle$ ($j=1,2,\ldots,r$).    
Then, Theorem~\ref{thm:state-transfer-simplified} guarantees that $v_r$ receives $\sum_s |s\rangle|A_r(s)\rangle$ and $\sum_s |s\rangle|B_r(s)\rangle$ with high fidelity 
as long as every node accepts with at least the probability guaranteed by Theorem~\ref{thm:state-transfer-simplified}. 
The SWAP test between these at $v_r$ checks if $A=B$ with high probability.

For the classical lower bound of $\SE$ in Theorem~\ref{thm:SetEQ_classical_lower_bound}, we utilize the lower bound for $\mathsf{EQ}_n^2$ of~\cite{FLNP21}. Ref.~\cite{FLNP21} showed that for any classical protocol for $\mathsf{EQ}_n^2$ on the line graph, at least one internal node requires a certificate of linear size. We show that $\mathsf{EQ}_n^2$ can be reduced to $\SE$ in three cases depending on the size of $U$. Here we explain the simplest case: $|U|=\Omega(\ell)$. For a line graph with the left-end extremity $v$ and the right-end extremity $v'$, let $x =x_1x_2\cdots x_n$ be the input of $\mathsf{EQ}_n^2$ for $v$ and $y =y_1y_2\cdots y_n$ be the input of $\mathsf{EQ}_n^2$ for $v'$. Then we consider an injection $f$ from $\{0,1\}^n$ to the set of $3\ell$-bit strings with Hamming weight $\ell$ such that the input list $(a_{v,1},\ldots,a_{v,\ell})$ of $\SE$ for $v$ includes the $j$-th element of the universal set $U$ for $|U|>3\ell$ if and only if $f(x)_j=1$, and the input list $(b_{v',1},\ldots,b_{v',\ell})$ of $\SE$ for $v'$ includes the $j$-th element of the universal set $U$ if and only if $f(y)_j=1$. Now these two sets are identical if and only if $x=y$, which means a reduction from $\mathsf{EQ}_n^2$ to $\SE$ for $\ell = \Theta(n)$. We thus get a lower bound of $\Omega(\ell)$ from the $\Omega(n)$ lower bound of $\mathsf{EQ}_n^2$ mentioned above.

The classical upper bound of $\SE$ in Theorem~\ref{thm:SetEQ_classical_upper_bound} is fairly simple: the prover can send all of inputs $A$ and $B$ to each node to achieve the first upper bound $O(r\ell \log |U|)$. For the second upper bound $O(|U| \log (r\ell))$, the node $v_i$ on the line graph $\{v_0,\ldots, v_r\}$ is given the information of inputs of $v_j, j\in\{0,\ldots,i-1\}$ as the certificate in the form of the number of each element of $U$ in the corresponding inputs.

%\ignore{
The basic proof idea of Theorem~\ref{thm:convertion-LOCC} is standard: 
we replace one qubit communicated between any two nodes $u$ and $v$ 
by two bits using quantum teleportation~\cite{BBC+93PRL}, 
assuming that they share an EPR pair $|\Phi^+\rangle=\frac{1}{\sqrt{2}}(|00\rangle+|11\rangle)$ sent from the prover.  The problem is that the prover may be malicious, and $u$ and $v$ should then verify that the pair sent from the prover is $|\Phi^+\rangle$. 
In order to obtain $|\Phi^+\rangle$ with high fidelity, we actually ask the prover to send $N+1$ copies of the EPR pairs. An honest prover will send the state $|\Phi^+\rangle^{\otimes (N+1)}$, but a malicious prover may naturally send an arbitrary state. Nodes $u$ and $v$ use $N$ among the $N+1$ pairs for the verification. If the verification succeeds, they are guaranteed that the remaining pair has high fidelity with $|\Phi^+\rangle$.  

This type of verification of $|\Phi^+\rangle$ in an adversarial scenario by the malicious prover was considered in a remarkable work by Zhu and Hayashi~\cite{ZH19PRA}. 
Extending the previous result~\cite{PLM18PRL} in a less adversarial scenario, they showed that by taking $N=O(\frac{1}{\varepsilon}\log(\frac{1}{\delta}))$, 
if the verification test succeeds with probability at least $\delta$, the state $\sigma$ of the last pair has a high fidelity with $|\Phi^+\rangle$ such that $\langle\Phi^+|\sigma|\Phi^+\rangle\geq 1-\varepsilon$. Furthermore, the measurements in their verification protocol (essentially the same as those in Ref.~\cite{PLM18PRL}) are local, namely, they do not need any entangled measurement between the two qubits of each pair.

Now the proof idea of Theorem~\ref{thm:convertion-LOCC} uses the verification protocol of Ref.~\cite{ZH19PRA} 
in our setting. To do so, we first observe that the amount of classical communication needed between $u$ and $v$ 
can be upper-bounded by $O(N)$ (which is the same as the certificate size from the prover), by rewriting the protocol of Ref.~\cite{ZH19PRA} with a slight modification in our setting. 
Then we replace the quantum bits sent among the nodes in the original dQMA protocol ${\cal P}$ by classical communication. However, it needs not only a single EPR pair but a lot of EPR pairs to be verified. 
Thus, we need further analysis to convert ${\cal P}$ into an LOCC dQMA protocol and to evaluate the message size of classical communication and the certificate size.  
%}

\subsection{Related work}\label{subsec:related work}

The concept of \textit{distributed Merlin-Arthur protocols} (dMA), which is very similar to the concept of \textit{randomized proof-labeling schemes}~\cite{FraigniaudPP19} was introduced by \cite{FraigniaudMORT19} as a randomized version of locally checkable proofs (LCPs). In a dMA protocol, as in LCPs, the prover assigns each node a short certificate.
The nodes then perform a 1-round distributed algorithm, i.e., exchange messages with their neighbors through incident edges. The difference is that in dMA, this algorithm can be a randomized algorithm, instead of a deterministic algorithm as in LCPs. This randomization is helpful to reduce the size of certificates for some problems.

The recent paper~\cite{KolOS18} introduced the interactive extension of dMA, \textit{distributed interactive proofs}, in which the prover and the verifier can perform more interaction. They showed that interaction is also useful to reduce the size of certificates. This concept has recently been explored in depth by several studies: distributed interactive proofs that utilize quantum certificates~\cite{LMN23}, the role of shared and private randomness~\cite{CrescenziFP19,MRR20}, and more efficient protocols for concrete problems~\cite{JMR22,MRR21,NPY20}.
In particular,~\cite{NPY20} introduced $\SE$, which is one of the problems we study in this paper, and showed efficient interactive protocols for $\SE$ when $\ell = |V|$ and $|U|=O(|V|)$ that require two interactions between the prover and the verifier with certificate size\footnote{For $\SE$, the certificate size of their protocol can be written as $O(\log|U|+\log (\ell |V|))$.} $O(\log |V|)$, and five interactions between the prover and the verifier with certificate size $O(\log \log |V|)$.

The technique we used in this paper from Refs.~\cite{MTH17PRA,ZH19PRA} belongs to a broad and hot topic called ``state certification (state verification)''~\cite{kliesch2021theory,yu2022statistical}. One conceptual contribution of this paper is providing the first concrete example of the effective use of these techniques for quantum distributed verification.

\section{Preliminaries}

\subsection{Quantum information}\label{sub:prelim}

We assume the familiarity with basics of quantum information such as quantum states, time evolutions, and measurements (see \cite{NC00,Wat18book,Wil17} for instance). 

For any quantum states $\sigma$ and $\rho$ in a Hilbert space ${\cal H}$, $D(\sigma,\rho)$ denotes the trace distance between $\sigma$ and $\rho$, 
namely, $D(\sigma,\rho):=\frac{1}{2}\|\sigma-\rho\|_1$, where $\|M\|_1:=\tr{(\sqrt{M^\dagger M})}$ is the trace norm of a matrix $M$. 
$F(\sigma,\rho)$ denotes the fidelity between $\sigma$ and $\rho$, 
namely, $F(\sigma,\rho):=\tr\sqrt{\sigma^{1/2}\rho\sigma^{1/2}}$. 
In particular, $F(|\psi\rangle\langle\psi|,\rho)^2=\langle\psi|\rho|\psi\rangle$, 
and $F(|\psi\rangle\langle\psi|,|\varphi\rangle\langle\varphi|)^2=|\langle\psi|\varphi\rangle|^2$.

The following three lemmas 
are used to evaluate how an ideal state we consider 
and the real state are close in this paper. 
The first is a well-known inequality 
between the trace distance and the fidelity 
called the Fuchs-van de Graaf inequalities 
(for instance, see~\cite{NC00,Wat18book,Wil17}). 
The second lemma can be found in Ref.~\cite{Wil17} for instance. 
The third lemma was proved in Ref.~\cite{MHNF15}.

\begin{lem}[Fuchs-van de Graaf inequalities]\label{lem:FvG}
For any $\sigma,\rho$ in a Hilbert space ${\cal H}$, 
\[
1-F(\sigma,\rho) \leq D(\sigma,\rho)\leq \sqrt{1-F(\sigma,\rho)^2}.
\]
\end{lem}

%The following lemma is simply called the union bound in Ref.~\cite{Wil17}.

%\ignore{

\begin{lem}[Union bound (for quantum measurements)]\label{lem:union-bound}
For any quantum state $\rho$ in a Hilbert space ${\cal H}$ 
and any two commuting projectors $\Pi_1$ and $\Pi_2$ on ${\cal H}$, 
\[
\tr[(I-\Pi_1\Pi_2)\rho] \leq \tr[(I-\Pi_1)\rho]+\tr[(I-\Pi_2)\rho].
\]
\end{lem}

\begin{lem}\label{lem:MHNF15}
Let $\rho$ be a state in $\mathcal{H}_1\otimes \mathcal{H}_2$, where $\mathcal{H}_1$ and $\mathcal{H}_2$ are Hilbert spaces.
For any $|x\rangle\in \mathcal{H}_1$,
%\begin{equation}\label{eq:MHNF15}
\[
\max_{\rho'\in\mathcal{H}_2} F(|x\rangle\langle x|\otimes\rho',\rho) = F(|x\rangle\langle x|,\tr_2(\rho)),
\]
%\end{equation} 
where $\tr_2$ is the partial trace over $\mathcal{H}_2$.
\end{lem}

%}

The SWAP test~\cite{BCWW01} is a quantum protocol 
to check the closeness of two given pure states 
$|\psi_1\rangle$ and $|\psi_2\rangle$, 
both of which are in a Hilbert space ${\cal H}$. 
Namely, it accepts with high probability 
when $|\psi_1\rangle$ and $|\psi_2\rangle$ are close. 
The following fact on the acceptance probability 
of the SWAP test is well-known.

\begin{lem}\label{lem:swap-acceptance-probability}
Given two (mixed) states $\sigma_1$ and $\sigma_2$ as input, 
the SWAP test accepts with probability 
$\frac{1}{2}+\frac{1}{2}\tr(\sigma_1\sigma_2)$. 
\end{lem}

Actually, the SWAP test receives 
not only the product state of two inputs 
$\sigma_1\in {\cal H}$ and $\sigma_2\in {\cal H}$ 
but any entangled state $\rho$ in ${\cal H}^{\otimes 2}$. 
The following lemma says that 
if the SWAP test accepts with high probability, 
the two reduced states of $\rho$ must be close~\cite{FLNP21}. 

\begin{lem}\label{lem:SWAPtest}
Let $z\geq 1$, and assume that the SWAP test on input $\rho$ in the input registers $({\sf R}_1,{\sf R}_2)$ accepts with probability $1-\frac{1}{z}$.
Then $D(\rho_1,\rho_2)\leq \frac{2}{\sqrt{z}} + \frac{1}{z}$, where $\rho_j$ is the reduced state on ${\sf R}_j$ of $\rho$. 
Moreover, if the SWAP test on input $\rho$ accepts with probability $1$, then $\rho_1=\rho_2$ (and hence $D(\rho_1,\rho_2)=0$).
\end{lem}  

For the measurements on a multipartite quantum system, 
several restricted classes of measurements are considered. 
The most famous one is $\mathrm{LOCC}$ (local operation and classical communication). 
A measurement on a $k$-partite system in $\mathrm{LOCC}$
consisting of subsystems $A_1,A_2,\ldots,A_k$ 
is implemented by local operations 
at each subsystem 
and classical communication among $k$ subsystems. 
A more restricted class is $\mathrm{LOCC}_1$ 
((fully) one-way LOCC) \cite{LS15PRL}. 
A measurement on a $k$-partite system in $\mathrm{LOCC}_1$
consisting of subsystems $A_1,A_2,\ldots,A_k$ 
is implemented by LOCC with the following order; 
local operation at $A_1$; for $j=1$ to $k-1$, 
classical communication from $A_j$ to $A_{j+1}$; 
local operation at $A_{j+1}$.

The quantum de Finetti theorems show 
that any state of $K$-partite system  
$A_1\cdots A_K$ that is a reduced state of a permutation-invariant state on $A_1\cdots A_N$ can be approximated 
by a mixture of $K$-fold products 
$\sigma^{\otimes K}$, if $N$ is sufficiently large than $K$. 
The following version 
of quantum de Finetti theorem 
is proved by Li and Smith~\cite{LS15PRL}, 
which has much better qualities 
on the size $d$ of each subsystem.

\begin{lem}[One-way LOCC measurement de Finetti theorem]\label{lemma:LS15}
Let $\rho_{{A}_1\cdots {A}_N}$ be a permutation-invariant state on ${\cal H}^{\otimes N}$, 
where ${A}_j$ is the $j$th subsystem over a $d$-dimensional system ${\cal H}$.
Then, for integer $0\leq K\leq N$, there exists a probabilistic measure $\mu$ on density matrices on ${\cal H}$ such that 
\[
\left\| 
\rho_{{A}_1\cdots {A}_K} - \int \sigma^{\otimes K} d\mu(\sigma) 
\right\|_{\mathrm{LOCC}_1}
\leq
\sqrt{ \frac{2(K-1)^2\ln d}{N-K} },
\]
where $\|\rho - \sigma\|_{\mathrm{LOCC}_1}
=\max_{M\in\mathrm{LOCC}_1} \| M(\rho)-M(\sigma) \|_1$
($M$ is the measurement operator corresponding to 
some POVM $\{M_x\}_x$\footnote{$M(\rho):=\sum_x \tr(\rho M_x)|x\rangle\langle x|$ 
with an orthogonal basis $\{|x\rangle\}_x$.}).
\end{lem}

\subsection{dQMA protocols}\label{sub:dQMA}

We consider a decision problem on a connected graph (called the network) $G=(V,E)$, where~$t$ inputs $x_1,x_2,\ldots,x_t$ are assigned to $t$ nodes $v_1,v_2,\ldots,v_t\in V$. We interpret the decision problem as a Boolean function $f$, where $f(x_1,x_2,\ldots,x_t)=1$ is interpreted as ``yes'' and $f(x_1,x_2,\ldots,x_t)=0$ is interpreted as ``no''.

%The goal is to compute $f(x_1,x_2,\ldots,x_t)$. 
%To compute $f(x_1,x_2,\ldots,x_t)$, 
%not only $v_1,v_2,\ldots,v_t$ but all nodes of $V$ 
%must communicate with one another in general.
%Communication is synchronous and is possible only if two nodes are connected: at each round, each node can send one message 

The concept of distributed quantum Merlin-Arthur (dQMA) protocols on a graph $G=(V,E)$ is a quantum version of the concept of distributed Merlin-Arthur (dMA) protocols. The aim of a dMA protocol is to verify whether $f(x_1,x_2,\ldots,x_t)=1$ or not. 
As briefly explained in Section \ref{sub:result1}, the nodes of $G$ (which correspond to the verifier) first receive a message from a powerful but possibly malicious party (the prover). The nodes then enter a verification phase, in which they communicate together (but do not communicate with the prover anymore). The communication is possible only if two nodes are connected: each node can send one message to each of its neighbors. In the case of dQMA protocols, the only difference is that the message from the prover and the communication among the nodes may be quantum. 
Note that neither randomness nor entanglement are shared among the nodes in advance. 

Formally, 
in a dQMA protocol ${\cal P}$ on $G=(V,E)$, 
each node $u\in V$ first receives a quantum register ${\sf M}_u$ from the prover.
Then the nodes move to the verification stage, which consists of the following steps:
(i) $u$ applies a local quantum (or classical) operation on the composite system of ${\sf M}_u$ and its private register ${\sf V}_u$; 
(ii) $u$ sends a quantum (or classical) register ${\sf M}_{uv}$ to any neighboring node $v$, and 
(iii) $u$ applies a local quantum (or classical) operation on ${\sf M}_u$, ${\sf V}_u$, and $\otimes_{v\in N(u)} {\sf M}_{vu}$, 
and either accepts or rejects (we call this the decision of $u$), 
where $N(u)$ denotes the set of nodes 
that are neighbors of $u$. When local operations at each node and communication among the nodes in the verification stage are classical, the dQMA protocol is called {\em LOCC (Local Operation and Classical Communication)}.

The two main complexity measures of ${\cal P}$ are the certificate size and the message size.
The certificate size of ${\cal P}$, denoted as $s_c^{{\cal P}}$, is the maximum number of qubits that are sent to each node from the prover, 
that is, 
$s_c^{{\cal P}}:=\max_{u\in V} |{\sf M}_u|$, where $|{\sf R}|$ denotes the number of qubits of ${\sf R}$. 
The message size of ${\cal P}$, denoted as $s_m^{{\cal P}}$, is the maximum number of qubits sent on edges of $G$, namely, 
$s_m^{{\cal P}}:=\max_{(u,v)\in E}(|{\sf M}_{uv}|+|{\sf M}_{vu}|)$. 

A dQMA protocol ${\cal P}$ 
for a decision problem $f$ on $G$  
with completeness $p_c$ and soundness $p_s$ 
is defined as a dQMA protocol satisfying 
the following two conditions:

\begin{description}
\item[(completeness)] 
\sloppy If $f(x_1,x_2,\ldots,x_t)=1$, 
there exists some quantum state 
$|\chi\rangle$ on 
${\sf M}:=\otimes_{u\in V} {\sf M}_u$ 
such that $\Pr[\mbox{all nodes accept}]\geq p_c$; 
\item[(soundness)]
If $f(x_1,x_2,\ldots,x_t)=0$,
for any quantum state $|\chi\rangle$ on ${\sf M}$, 
$\Pr[\mbox{all nodes accept}]\leq p_s$.
\end{description}

In this paper, we consider the problem of generating 
a quantum state $|\varphi\rangle$ on a network $G=(V,E)$. 
In this problem, some initially specified nodes $w_1,\ldots,w_{\kappa}$ not only make their decisions (accept or reject) but also output the quantum state $|\varphi\rangle$ jointly (if they accept). 
In our specific problem, the $n$-qubit $\SGDI$, 
%on the line of length $r$, 
all nodes of the line graph with nodes 
$v_0,v_1,\ldots,v_r$ have an input
($v_0$ has a classical description of $|\psi\rangle$ and
$v_j$ for $j=1,2,\ldots,r$ has a classical description of $U_j$),
$|\varphi\rangle=|\varphi_r\rangle~(:=U_r\cdots U_1|\psi\rangle)$, 
$\kappa=1$, and $w_1=v_r$. 

In a dQMA protocol for the problem of generating $|\varphi\rangle$ on $G$, 
the completeness and soundness conditions 
are slightly different from the case of decision problems. For our purpose we actually only need to discuss perfect-completeness protocols.
We say that the dQMA protocol has perfect completeness 
and $(\delta,\varepsilon)$-soundness 
if the following completeness and soundness are satisfied:

\begin{description}
\item[(completeness)] 
There exists a quantum state $|\chi\rangle$ on 
${\sf M}$ such that 
$$
\Pr[\mbox{all nodes accept and $w_1,\ldots,w_{\kappa}$ output $|\varphi\rangle$ jointly}]=1;
$$
\item[(soundness)]
If all nodes accept with probability at least $\delta$,
then the output $\tilde{\rho}$ of $w_1,\ldots,w_{\kappa}$ (under the condition that all nodes accept) satisfies
\[
\langle\varphi|\tilde{\rho}|\varphi\rangle \geq 1-\varepsilon.
\]
\end{description}

The soundness condition is regarded as 
a kind of hypothesis testing (i.e., if the verifier's test passes with probability 
greater than a threshold, then the state would be close to the ideal one). A similar completeness-soundness condition is used for the interactive proofs for synthesizing quantum states \cite{RY22}.

%======================================================
\section{dQMA Protocol for State Generation with Distributed Inputs}\label{sec:SDGI}
%\section{Proof of Theorem~\ref{thm:state-transfer-simplified}}
In this section we present our dQMA protocol for the $n$-qubit State Generation with Distributed Inputs over the line of length $r$ ($n$-qubit $\SGDI$) and prove Theorem~\ref{thm:state-transfer-simplified}.

\ignore{The parties $v_0,v_1,\ldots,v_r$ are the nodes of the line graph of length $r$, where $v_j$ and $v_{j+1}$ are connected $(j=0,1,\ldots,r-1)$. The left-end extremity $v_0$ has an $n$-qubit state $|\psi\rangle$, which can be prepared by $v_0$, as input. $v_j$ $(j=1,2,\ldots,r)$ has an $n$-qubit unitary transformation $U_j$ as input. In this setting, the aim is to generate the quantum state\[|\varphi_r\rangle:=U_{r}\cdots U_1|\psi\rangle\] at the right-end extremity $v_r$. We call this problem {\em State Generation with Distributed Inputs (SGDI)} on the line of length $r$. This problem is not solvable in one round clearly by considering the case where $U_1=\cdots=U_r=I$.    

Instead, the prover, who may be malicious, helps them in a non-interactive way, that is, he/she just sends a quantum message to them. In this case, we can verify the generate of $|\varphi_r\rangle$ at $v_r$ in one round.  
}

%\subsection{Result}

\subsection{dQMA protocol for SGDI}

The following is our dQMA protocol for $n$-qubit $\SGDI$.

\fboxsep=6pt
\begin{breakbox}
\noindent
{\bf Protocol ${\cal P}_{{\sf SGDI}}$}: 
Let $k=144cr^{2+\eta}$ and $m=2cn k^2(r+1)^{1+\eta}$ 
for any constant $c>0$ 
and any small constant $\eta\geq 0$.

\begin{enumerate}
\item $v_0$ prepares $(m+k+1)$ copies of $|\psi\rangle$ in $n$-qubit registers ${\sf R}_{0,j}$ ($j=1,2,\ldots,m+k+1$).
\item The prover sends each $v_{l}$, where $l=1,2,\ldots,r$, $(m+k+1)$ $n$-qubit registers ${\sf R}_{l,1},{\sf R}_{l,2},\ldots,{\sf R}_{l,m+k+1}$.
\item Each $v_l$ ($l=1,2,\ldots,r$) permutes the $(m+k+1)$ registers ${\sf R}_{l,1},{\sf R}_{l,2},\ldots,{\sf R}_{l,m+k+1}$ 
by a permutation $\pi$ on $\{1,2,\ldots,m+k+1\}$ taken uniformly at random, and renames ${\sf R}_{l,j}:={\sf R}_{l,\pi(j)}$. 
\item \sloppy The parties $v_0,v_1,\ldots,v_r$ implement the following subprotocol ${\cal P}_{{\sf SGDIV}}$ (a modification of the verification steps in Ref.~\cite{FLNP21}) 
on registers ${\sf R}_{0,j},{\sf R}_{1,j},\ldots,{\sf R}_{r,j}$ for each $j=2,3,\ldots,k+1$ in order.
If some party rejects for some $j$, the protocol rejects.
\item $v_r$ outputs ${\sf R}_{r,1}$.
\end{enumerate}
\end{breakbox}\vspace{3mm}

\begin{breakbox}
\noindent
{\bf Protocol ${\cal P}_{{\sf SGDIV}}$}: 
Assume that $v_0$ has $|\psi\rangle$ on $n$-qubit register ${\sf R}_0$, 
and $v_{l}$ ($l=1,2,\ldots,r$) receives $n$-qubit register ${\sf R}_{l}$.  

\begin{enumerate}
\item For every $j=0,1,\ldots,r-1$, party $v_j$ chooses a bit $b_j$ uniformly at random, 
and sends its register ${\sf R}_j$ to the right neighbor $v_{j+1}$ whenever $b_j=0$.
\item For every $j=1,2,\ldots,r$, if $v_j$ receives a register from the left neighbor $v_{j-1}$, 
and if $b_j=1$, then $v_j$ applies $U_j$ on register ${\sf R}_{j-1}$, and performs the SWAP test 
on the registers $({\sf R}_{j-1},{\sf R}_j)$, and accepts or rejects accordingly; Otherwise, $v_j$ accepts. 
\end{enumerate}
\end{breakbox}\vspace{3mm}
\fboxsep=3pt

We can show the following theorem, 
which induces Theorem~\ref{thm:state-transfer-simplified} 
by a special case with $\eta=0$.

\begin{theorem}\label{thm:state-transfer}
Protocol ${\cal P}_{{\sf SGDI}}$ has 
perfect completeness and 
$(\frac{1}{(cr^{\eta})^{1/4}},\frac{1}{(cr^{\eta})^{1/4}})$-soundness. 
The certificate size of ${\cal P}_{{\sf SGDI}}$ is $O(n^2r^{5+3\eta})$ and the message size is $O(nr^{2+\eta})$.
\ignore{the following property: 
%There exists a dQMA protocol for SGDI on the line of length $r$ with the following property:
\begin{description}
\item[(completeness)] 
If the prover is honest, 
that is, the contents of ${\sf R}_{l,j}$ is $$
|\varphi_{l}\rangle := U_{l}\cdots U_2U_1|\psi\rangle
$$ for any $l\in\{1,2,\ldots,r\}$ and $j\in\{1,2,\ldots,m+k+1\}$, 
$v_r$ outputs $|\varphi_r\rangle$ with probability $1$.  
\item[(soundness)]
If ${\cal P}_{{\sf SGDI}}$ accepts (i.e., no one rejects) with probability at least $\frac{1}{(cr^{\eta})^{1/4}}$,  
%$v_r$ outputs ${\sf R}_{r,1}$ without rejecting the protocol, 
then the contents $\rho$ of ${\sf R}_{r,1}$ satisfies 
\[
\langle \varphi_r|\rho|\varphi_r \rangle \geq 1-\frac{1}{(cr^{\eta})^{1/4}}.
\]
\end{description}
}
\end{theorem}

\subsection{Proof of Theorem~\ref{thm:state-transfer}}
  
We can see that the certificate size of ${\cal P}_{{\sf SGDI}}$ is $(m+k+1)n=O(n^2 r^{5+3\eta})$ 
from step 2 of ${\cal P}_{{\sf SGDI}}$. Since ${\cal P}_{{\sf SGDI}}$ implements ${\cal P}_{{\sf SGDIV}}$ $(k+1)$ times,
and the message size of ${\cal P}_{{\sf SGDI}}$ is $O(n)$, the message size of ${\cal P}_{{\sf SGDI}}$ is $O(nk)=O(nr^{2+\eta})$.

The completeness clearly holds: since the prover honestly sends 
\[
|\varphi_l\rangle:=U_l\cdots U_1|\psi\rangle 
\]
as the content of ${\sf R}_{l,j}$ 
for each $j\in\{1,2,\ldots,m+k+1\}$ 
and then all the SWAP tests in ${\cal P}_{{\sf SGDIV}}$ accept with probability $1$. 

%The proof of the soundness can be found in Appendix~\ref{Appendix:proof-lem:FLNP21}.

%In this section, we complete the proof of Theorem~\ref{thm:state-transfer} by proving the soundness of ${\cal P}_{{\sf SGDI}}$.
For the soundness, we use the following lemma on the subprotocol ${\cal P}_{{\sf SGDIV}}$. 
  
\begin{lem}\label{lem:FLNP21}
Let $\rho$ be the reduced state from the prover to $v_r$ before the local test in step 2 of ${\cal P}_{{\sf SGDIV}}$. 
If $\langle \varphi_r|\rho|\varphi_r\rangle \leq 1-\frac{1}{\alpha(\rho)}$, 
the probability that some party rejects is at least $\beta(\rho)=\frac{1}{72\alpha(\rho)^2 r^2}$.
\end{lem}

\begin{proof}
Let $\rho_0=|\psi\rangle\langle\psi|$, and let $\rho_j$ be the reduced state of ${\sf R}_j$ before the local test in step 2 of ${\cal P}_{{\sf SGDIV}}$, 
for $j=1,2,\ldots,r$ (note that $\rho_r=\rho$).
For every $j=1,2,\ldots,r$, let $F_j$ be the event that $v_j$ performs the local test in step 2 of ${\cal P}_{{\sf SGDIV}}$, 
and let $E_j$ be the event that the local test rejects. 

Let $\alpha_j = \Pr[E_j|F_j]$. 
Then, for every $j=1,2,\ldots,r$, $\Pr[\overline{E_j}|F_j]=1-\alpha_j$, 
where we note that the complementary event $\overline{E_j}$ is the event 
where the SWAP test on the two $n$-qubit states on ${\sf R}_{j-1}$ after applying $U_j$ and ${\sf R}_j$ accepts.
By Lemma~\ref{lem:SWAPtest},
\[
D(U_j\rho_{j-1}U_j^\dagger, \rho_j)\leq \left\{
\begin{array}{ll}
\frac{2}{\sqrt{1/\alpha_j}} + \frac{1}{1/\alpha_j} & \alpha_j\neq 0\\
0 & \mbox{otherwise}
\end{array}
\right.
\]
and thus 
\begin{equation}\label{eq:dist-bound-base}
D(U_j\rho_{j-1}U_j^\dagger, \rho_j)\leq 3\sqrt{\alpha_j}.
\end{equation} 
Then the following fact is shown by induction.

\

\noindent
%{\bf Fact 1}: 
\begin{fact}\label{fact1}
For any $\nu\in\{1,2,\ldots,r\}$, 
\[
D(|\varphi_{\nu}\rangle\langle\varphi_{\nu}|, \rho_{\nu})\leq 3\sum_{j=1}^{\nu} \sqrt{\alpha_j}.
\]
\end{fact}
\noindent
\begin{proof}[Proof of Fact \ref{fact1}]
%{\bf Proof of Fact 1}: 
$\nu=1$ holds by the inequality~(\ref{eq:dist-bound-base}) (the $j=1$ case).
For $\nu>1$, 
\begin{align*}
D(|\varphi_{\nu}\rangle\langle\varphi_{\nu}|, \rho_{\nu})
&\leq D(|\varphi_{\nu}\rangle\langle\varphi_{\nu}|, U_{\nu}\rho_{\nu-1}U_{\nu}^{\dagger}) + D(U_{\nu}\rho_{\nu-1}U_{\nu}^{\dagger}, \rho_{\nu})\\
&= D(|\varphi_{\nu-1}\rangle\langle\varphi_{\nu-1}|, \rho_{\nu-1}) + D(U_{\nu}\rho_{\nu-1}U_{\nu}^{\dagger}, \rho_{\nu})\\
&\leq 3\sum_{j=1}^{\nu-1} \sqrt{\alpha_j} + 3\sqrt{\alpha_{\nu}}\\
&= 3\sum_{j=1}^{\nu} \sqrt{\alpha_j},
\end{align*}
where the second inequality uses the induction and the inequality~(\ref{eq:dist-bound-base}). 
\end{proof}

\

By the Cauchy-Schwarz inequality and Fact 1,
\[
D(|\varphi_r\rangle\langle\varphi_r|,\rho)
=
D(|\varphi_r\rangle\langle\varphi_r|,\rho_r)
\leq 
3\sum_{j=1}^r \sqrt{\alpha_j }
\leq 3\sqrt{r}\sqrt{\sum_{j=1}^r \alpha_j}.
\]
On the contrary, by Lemma~\ref{lem:FvG},
\[
D(|\varphi_r\rangle\langle\varphi_r|, \rho)
\geq 1 - F(|\varphi_r\rangle\langle\varphi_r|,\rho)
\geq 1- \sqrt{1-\frac{1}{\alpha(\rho)}}
\geq 1- \left(1-\frac{1}{2\alpha(\rho)}\right)=\frac{1}{2\alpha(\rho)},
\]
where the second inequality comes from 
the assumption $\langle\varphi_r|\rho|\varphi_r\rangle\leq 1-\frac{1}{\alpha(\rho)}$, 
and the third inequality comes from 
$\sqrt{1-x}\leq 1-\frac{x}{2}$ 
($0\leq x\leq 1$). 
Thus, 
\[
\sum_{j=1}^r \alpha_j \geq \frac{1}{36\alpha(\rho)^2 r}.
\]
%we obtain 
%\[
%\sqrt{r}\sqrt{\sum_{j=1}^r \alpha_j} \geq \sum_{j=1}^r \sqrt{\alpha_j} \geq \frac{1}{3}(\alpha_r+3\sum_{j=1}^{r-1}\sqrt{\alpha_j})\geq \frac{1}{6\alpha(\rho)}.
%\]
Therefore, 
if $\langle \varphi_r|\rho|\varphi_r\rangle \leq 1-\frac{1}{\alpha(\rho)}$, 
\begin{equation}\label{eq:0830}
\sum_{j=1}^r \Pr[E_j|F_j]\geq \frac{1}{36\alpha(\rho)^2 r}.
\end{equation}

In step 2 (of ${\cal P}_{{\sf SGDIV}}$), $v_j$ ($j=1,2,\ldots,r$) performs the local test with probability at least $1/4$.
It follows that, for every $j=1,2,\ldots,r$, the event $F_j$ occurs in at least $(1/4)\times 2^r$ outcomes of all the $2^r$ possible outcomes $b_0\cdots b_{r-1}$ 
that induce $\kappa$ events $F_{j_1},\ldots,F_{j_{\kappa}}$ with $\kappa \neq 0$ where we note that $0\leq \kappa\leq \lfloor r/2\rfloor$ in general.
The probability that some node rejects in step 2 {\em under this outcome} is
\[
\Pr[\vee_{i=1}^{\kappa} E_{j_i}| \wedge_{i=1}^{\kappa} F_{j_i}]
\geq 
\frac{1}{\lfloor r/2\rfloor} \sum_{i=1}^{\kappa} \Pr[E_{j_i}| \wedge_{i=1}^{\kappa} F_{j_i}] 
= \frac{1}{\lfloor r/2\rfloor} \sum_{i=1}^{\kappa} \Pr[E_{j_i}|F_{j_i}],
\] 
where the first inequality comes from the elementary inequality 
on probability $\Pr[\vee_{j=1}^{\nu} A_j]\geq \frac{1}{\nu}\sum_{j=1}^{\nu} \Pr[A_j]$, 
and the equality comes from the fact that each of $F_{j_i}$ and $E_{j_i}$ is independent from all the other event $F_{j_{i'}}$ with $i'\neq i$
(note that $|j_{i'}-j_i|\geq 2$ since $F_{j-1}$ and $F_j$ never occur at the same time).
As each outcome occurs with probability $\frac{1}{2^r}$, the probability that some node rejects in step 2 is at least
\[
\frac{1}{2^r} \cdot [(1/4)\cdot 2^r]\cdot \frac{1}{\lfloor r/2\rfloor} \sum_{j=1}^r \Pr[E_{j}|F_{j}]\geq \frac{1}{2r} \sum_{j=1}^r \Pr[E_{j}|F_{j}]
\geq \frac{1}{72\alpha(\rho)^2 r^2},
\]
where the last inequality comes from the inequality~(\ref{eq:0830}).
This completes the proof of Lemma~\ref{lem:FLNP21}.
\end{proof}

The remaining arguments are based on the idea of the analysis in Ref.~\cite{MTH17PRA}. 

Fix $j\in\{1,2,\ldots,k+1\}$ arbitrarily. 
Let $\rho$ be the $v_r$'s contents in ${\sf R}_{r,j}$ before running ${\cal P}_{{\sf SGDIV}}$ (step 4 in ${\cal P}_{{\sf SGDI}}$), 
and let 
\begin{equation}\label{eq:1}
\langle\psi|\rho|\psi\rangle=1-\frac{1}{\alpha(\rho)}.
\end{equation}
Let $\sigma$ be the $n(r+1)$-qubit reduced state 
on ${\sf R}_{*j}:=\otimes_{l=0}^r {\sf R}_{l,j}$ 
before running ${\cal P}_{{\sf SGDIV}}$.  
%${\sf R}_{0,j},{\sf R}_{1,j},\ldots,{\sf R}_{r,j}$.   %for a arbitrary fixed $j\in\{1,2,\ldots,k+1\}$. 
Thus $\rho=\tr_{{\sf B}_j} (\sigma)$, 
where ${\sf B}_j$ is the system of $v_0,v_1,\ldots,v_{r-1}$, that is, ${\sf B}_j:=\otimes_{l=0}^{r-1} {\sf R}_{l,j}$. 
Let 
$$
\Pi^\bot=I_{{\sf B}_j} \otimes (I-|\psi\rangle\langle\psi|)_{{\sf R}_{r,j}}.
$$ 
Then, Eq.~(\ref{eq:1}) means that 
\begin{equation}\label{eq:2}
\tr(\Pi^\bot \sigma) = 1-\langle \psi|\rho|\psi\rangle = \frac{1}{\alpha(\rho)}.
\end{equation}

Let $\{M,I-M\}$ be the binary POVM such that $M$ corresponds to the acceptance (i.e., the case where no one rejects) 
in the verification of ${\cal P}_{{\sf SGDIV}}$. 
By Lemma~\ref{lem:FLNP21} and Eq.~(\ref{eq:2}),
\begin{equation}\label{eq:3}
\tr\left[ ( \Pi^\bot \otimes M^{\otimes k}  )  \sigma^{\otimes (k+1)} \right] 
\leq (1-\beta(\rho))^k\cdot \frac{1}{\alpha(\rho)}.
\end{equation} 
The righthand of the inequality~(\ref{eq:3}) is a function of $\frac{1}{\alpha(\rho)}$, 
and achieves the maximum value 
\begin{equation}\label{eq:3-1}
\left( 1-\frac{1}{2k+1} \right)^k \sqrt{\frac{72r^2}{2k+1}} \leq \sqrt{\frac{72r^2}{2k+1}} \leq \frac{1}{ 2(cr^{\eta})^{1/2} }
\end{equation}
when $\frac{1}{\alpha(\rho)} = \sqrt{\frac{72r^2}{2k+1}}$ (recalling that $k=144cr^{2+\eta}$).  

Now let $\xi$ be the state on $\otimes_{l=0}^r\otimes_{j=1}^{k+1} {\sf R}_{l,j}$ before running ${\cal P}_{{\sf SGDIV}}$. 
%with $l\in\{0,1,\ldots,r\}$ and $j\in\{1,2,\ldots,k+1\}$.
By the one-way LOCC measurement de Finetti theorem 
(Theorem \ref{lemma:LS15}) and the inequality~(\ref{eq:3-1}),
\begin{align}
\tr\left[ ( \Pi^\bot \otimes M^{\otimes k} )  \xi \right]
&\leq 
\int d\mu( \sigma ) 
\tr\left[ (\Pi^\bot \otimes M^{\otimes k}  )  \sigma^{\otimes (k+1)} \right] + \frac{1}{2}\sqrt{\frac{2k^2\ln (2^{n(r+1)})}{m} }\nonumber\\
&\leq
\frac{1}{2(cr^{\eta})^{1/2} } + \frac{1}{2}\sqrt{\frac{2k^2n(r+1)}{m} }\nonumber\\
&\leq \frac{1}{ (cr^{\eta})^{1/2} } \label{eq:4}. 
\end{align}
(Recall $m=2cn k^2(r+1)^{1+\eta}$.) 
%the righthand of the inequality~(\ref{eq:4}) is at most $\frac{1}{ (cr^{\eta})^{1/2} }$.

\ignore{Here, let \[\xi=\frac{(I\otimes \sqrt{M}^{\otimes k})\Psi(I\otimes \sqrt{M}^{\otimes k})}{\tr[(I\otimes M^{\otimes k})\Psi]}, <-- State after POVM is not uniquely determined!!!\]　where the identity $I$ acts on ${\sf R}_1:=\otimes_{j=0}^r {\sf R}_{j,1}$, and $\tilde{\sigma}$ be the reduced state of $\xi$ on ${\sf R}_1$. 
}
Then we have 
%Note that 
\begin{equation}\label{eq:5}
\tr\left[ ( \Pi^\bot \otimes M^{\otimes k} )  \xi \right]
%=
%\tr[(\Pi^\bot\otimes I)\xi]\tr\left[ ( I \otimes M^{\otimes k} ) \Psi \right]
=
\tr\left[ \Pi^\bot \tilde{\sigma} \right] 
\tr\left[ ( I \otimes M^{\otimes k} ) \xi \right],
\end{equation}
where $\tilde{\sigma}$ is 
the state on ${\sf R}_{*1}=
\otimes_{l=0}^r {\sf R}_{l,1}$ 
under the condition that the protocol accepts (i.e., no one rejects). 
By the inequality~(\ref{eq:4}) and Eq.~(\ref{eq:5}), 
if $\tr\left[ \Pi^\bot \tilde{\sigma} \right]
> \frac{1}{(cr^{\eta})^{1/4}}$, \sloppy 
then 
$\tr\left[ (I \otimes M^{\otimes k} ) \xi \right] < \frac{1}{(cr^{\eta})^{1/4}}$.

Let $\tilde{\rho}$ be the reduced state of $\tilde{\sigma}$ on ${\sf R}_{r,1}$. 
Note that $\tilde{\rho}$ 
is the output state of node $v_r$ under the condition that the protocol accepts. 
Since 
$
\langle\psi|\tilde{\rho}|\psi\rangle
=
1-\tr[\Pi^\bot \tilde{\sigma}],
$
%(recalling that $\rho=\tr_{{\sf B}}(\sigma)$ and $\langle \psi|\rho|\psi\rangle=1-\tr\left[ \Pi^\bot \sigma \right]$)
if the protocol accepts with probability $\geq \frac{1}{ (cr^{\eta})^{1/4} }$, 
then %with probability $1-\frac{1}{r^{1/4}}$,
\[
\langle\psi|\tilde{\rho}|\psi\rangle\geq 1-\frac{1}{ (cr^{\eta})^{1/4} }.
\]
Now the soundness proof of ${\cal P}_{{\sf SGDI}}$, and thus, 
the proof of Theorem~\ref{thm:state-transfer} are completed.

%=====================================================
\section{Application: dQMA Protocol for Set Equality}

In this section we prove Theorem~\ref{thm:SetEQ} by constructing a protocol for $\SE$ based on the protocol for $\SGDI$ developed in Section \ref{sec:SDGI}.

\begin{proof}[Proof of Theorem~\ref{thm:SetEQ}]
We consider $\SE$ (Definition~\ref{def:seteq}) for the line graph of length $r$ with nodes $v_0,v_1,\ldots,v_r$, 
where $v_j$ has $a_{j,1},\ldots,a_{j,\ell }$ and $b_{j,1},\ldots,b_{j,\ell }$.
Let 
$\alpha_j(s):= \prod_{i\in\{1,2,\ldots,\ell \}} (s-a_{j,i})$ and 
$\beta_j(s):=\prod_{i\in\{1,2,\ldots,\ell \}} (s-b_{j,i})$ for each $j\in\{0,1,\ldots,r\}$. 
We identify $a_{u,i},b_{u,i}$ as elements in a finite field $\mathbb{F}$ with size $|\mathbb{F}|\geq \tilde{c}\ell (r+1)2^{\log|U|}$ for some (sufficiently large) constant $\tilde{c}>0$.

Our protocol is as follows:
%\
\fboxsep=6pt
\begin{breakbox}
\noindent
{\bf Protocol ${\cal P}_{{\sf seteq}}$}:

\begin{enumerate}
\item $v_0,v_1,\ldots,v_r$ implement the protocol ${\cal P}_{{\sf SGDI}}$ of Theorem~\ref{thm:state-transfer} 
with $n=2(2\lceil\log |\mathbb{F}|\rceil+\lceil\log (r+1)\rceil)$
and with a sufficiently large $c>0$ and $\eta=0$
for the following case.

\begin{itemize}
\item Let $|\psi\rangle
=|\psi_A\rangle\otimes|\psi_B\rangle$, 
where   
\[
|\psi_A\rangle
=
\frac{1}{\sqrt{|\mathbb{F}|} } \sum_{s\in \mathbb{F}} 
|s\rangle_{{\sf R}_{A,1}} |\alpha_0(s)\rangle_{{\sf R}_{A,2}} 
|0\rangle_{{\sf R}_{A,3}}
\]
and 
\[
|\psi_B\rangle=
\frac{1}{\sqrt{|\mathbb{F}|} } \sum_{s\in \mathbb{F}} 
|s\rangle_{{\sf R}_{B,1}} |\beta_0(s)\rangle_{{\sf R}_{B,2}} 
|0\rangle_{{\sf R}_{B,3}},
\] 
where the contents of ${\sf R}_{A,1},{\sf R}_{A,2},{\sf R}_{B,1},{\sf R}_{B,2}$ are elements of $\mathbb{F}$, 
and the contents of ${\sf R}_{A,3}$ and ${\sf R}_{B,3}$ 
are numbers from $\{0,1,\ldots,r\}$. 
\item Let $U_j$ be
\[
U_j=G_{j,A}\otimes G_{j,B}.
\]
Here, $G_{j,A}$ (depending on $a_{j,1},\ldots,a_{j,\ell }$) 
is defined as a unitary transformation satisfying 
(A1) $G_{j,A}(|s\rangle|t\rangle|0\rangle):=|s\rangle|\alpha_j(s) t\rangle|0\rangle$ if $\alpha_j(s)\neq 0$, 
(A2) $G_{j,A}(|s\rangle|t\rangle|0\rangle):=|s\rangle|t\rangle|1\rangle$ if $\alpha_j(s)= 0$, 
(A3) $G_{j,A}(|s\rangle|t\rangle|\nu\rangle):=|s\rangle|t\rangle|\nu+1\rangle$ if $\nu\in\{1,2,\ldots,r-1\}$, and 
$G_{j,B}$ (depending on $b_{j,1},\ldots,b_{j,\ell }$) 
is defined as a unitary transformation satisfying 
(B1) $G_{j,B}(|s\rangle|t\rangle|0\rangle):=|s\rangle|\beta_j(s) t\rangle|0\rangle$ if $\beta_j(s)\neq 0$, 
(B2) $G_{j,B}(|s\rangle|t\rangle|0\rangle):=|s\rangle|t\rangle|1\rangle$ if $\beta_j(s)= 0$, 
(B3) $G_{j,B}(|s\rangle|t\rangle|\nu\rangle):=|s\rangle|t\rangle|\nu+1\rangle$ if $\nu\in\{1,2,\ldots,r-1\}$.
\end{itemize}

\item Node $v_r$ does the SWAP test for the output system 
${\sf R}_A\otimes {\sf R}_B$ of ${\cal P}_{{\sf SGDI}}$, 
and accepts or rejects accordingly.  
\end{enumerate}
\end{breakbox} \vspace{3mm} \par
\fboxsep=3pt
\noindent\textbf{Analysis:} First, we can see that 
the certificate size and the message size of ${\cal P}_{{\sf seteq}}$
are as desired from Theorem~\ref{thm:state-transfer} with $n=O(\log(\ell r)\log|U|)$. 

Now we consider the completeness and soundness of the protocol ${\cal P}_{{\sf seteq}}$. 
Let  
\[
p_A(s):=\prod_{j\in\{0,1,\ldots,r\} } \alpha_j(s) = \prod_{j\in\{0,1,\ldots,r\},\ i\in\{1,2,\ldots,\ell \}} (a_{j,i}-s)
\] 
and  
\[
p_B(s):=\prod_{j\in\{0,1,\ldots,r\} } \beta_j(s) = \prod_{j\in\{0,1,\ldots,r\},\ i\in\{1,2,\ldots,\ell \}} (b_{j,i}-s).
\]
Note that if $A=B$, $p_A(s)=p_B(s)$ for all $s\in \mathbb{F}$, while 
if $A\neq B$, there are at most $\ell (r+1)$ elements $s$ such that $p_A(s)=p_B(s)$ 
since $p_A(s)-p_B(s)$ is not zero-polynomial and its degree is at most $\ell (r+1)$. \vspace{2mm}
 
\noindent
{\bf Completeness:} The prover does honest operations for the run of ${\cal P}_{{\sf SGDI}}$.
Then, node $v_r$ obtains 
$|\psi_A^f\rangle\otimes |\psi_B^f\rangle$ 
as the output of ${\cal P}_{{\sf SGDI}}$, 
where
%\begin{equation}\label{eq:r-psiA}
\[
|\psi_A^f\rangle:=G_{r,A}\cdots G_{1,A}|\psi_A\rangle
\]
%\end{equation}
and 
%\begin{equation}\label{eq:r-psiB}
\[
|\psi_B^f\rangle:=G_{r,B}\cdots G_{1,B}|\psi_B\rangle.
\]
%\end{equation}

Let $F_A:=\{s\mid p_A(s)=0\}$ and $F_B:=\{s\mid p_B(s)=0\}$. 
Note that $|F_A|\leq \ell (r+1)$ and $|F_B|\leq \ell (r+1)$, and hence $|\overline{F_A}\cap \overline{F_B}|\geq |\mathbb{F}|-|F_A|-|F_B|\geq |\mathbb{F}|-2\ell (r+1)$.

By definition,
\begin{equation}\label{eq:r-psiA+}
|\psi_A^f\rangle = \frac{1}{\sqrt{|\mathbb{F}|}} \sum_{s\in \overline{F_A}} |s\rangle|p_A(s)\rangle|0\rangle + \frac{1}{\sqrt{|\mathbb{F}|}} \sum_{s\in F_A} |s\rangle|g_A(s)\rangle|h_A(s)\rangle 
\end{equation}
for some functions $g_A$ and $h_A$ satisfying $h_A(s)\neq 0$ for any $s\in F_A$. Similarly,
\begin{equation}\label{eq:r-psiB+}
|\psi_B^f\rangle = \frac{1}{\sqrt{|\mathbb{F}|}} \sum_{s\in \overline{F_B}} |s\rangle|p_B(s)\rangle|0\rangle + \frac{1}{\sqrt{|\mathbb{F}|}} \sum_{s\in F_B} |s\rangle|g_B(s)\rangle|h_B(s)\rangle 
\end{equation}
for some functions $g_B$ and $h_B$ satisfying $h_B(s)\neq 0$ for any $s\in F_B$.
Since the coefficients of $|\psi_A^f\rangle$ and $|\psi_B^f\rangle$ are nonnegative,
%\begin{align*}
\[
\langle\psi_A^f|\psi_B^f\rangle 
\geq \frac{1}{|\mathbb{F}|} \sum_{s\in \overline{F_A}\cap \overline{F_B} } 1
\geq \frac{|\mathbb{F}|-2\ell (r+1)}{|\mathbb{F}|}=1-\frac{2\ell (r+1)}{|\mathbb{F}|}.
\]
%\end{align*}
Thus, by Lemma~\ref{lem:swap-acceptance-probability}, the SWAP test of $v_r$ accepts with probability 
\[
\frac{1}{2}+\frac{1}{2}|\langle\psi_A^f|\psi_B^f\rangle|^2\geq 1-\frac{2\ell (r+1)}{|\mathbb{F}|},
\]
which is at least $1-1/c'$ where $c'$ is any large constant when we appropriately choose the constant $\tilde{c}$.\vspace{2mm}

\noindent
{\bf Soundness:} 
First we consider the ideal case that 
node $v_r$ obtains 
$|\psi_A^f\rangle\otimes |\psi_B^f\rangle$ 
as the output of ${\cal P}_{{\sf SGDI}}$. 
By Eq.~(\ref{eq:r-psiA+}) and Eq.~(\ref{eq:r-psiB+}),
\begin{align*}
\langle\psi_A^f|\psi_B^f\rangle
&=\frac{1}{|\mathbb{F}|} \sum_{s\in\overline{F_A}\cap\overline{F_B}} \langle p_A(s)|p_B(s)\rangle + \frac{1}{|\mathbb{F}|} \sum_{s\in F_A\cap F_B} \langle g_A(s)|g_B(s)\rangle\langle h_A(s)|h_B(s)\rangle\\
&\leq \frac{1}{|\mathbb{F}|} \cdot \ell (r+1) + \frac{1}{|\mathbb{F}|} |F_A\cap F_B|\\
&\leq \frac{2\ell (r+1)}{|\mathbb{F}|},
\end{align*}
where the first inequality comes from the fact that if $A\neq B$, 
there are at most $\ell (r+1)$ elements $s$ such that $p_A(s)=p_B(s)$, 
and the second inequality comes from $|F_A\cap F_B|\leq |F_A|\leq \ell (r+1)$. 
Then the SWAP test of $v_r$ accepts with probability at most
\[
\frac{1}{2} + \frac{1}{2} \left(\frac{2\ell (r+1)}{|\mathbb{F}|}\right)^2 = \frac{1}{2} + 2 \left(\frac{\ell (r+1)}{|\mathbb{F}|}\right)^2.  
\]

Next, we consider the case that 
${\cal P}_{{\sf SGDI}}$ accepts with probability less than $\frac{1}{c^{1/4}}$.
This means the verifier accepts with only this low probability, and the soundness holds in this case.

Finally, 
we consider the case that 
${\cal P}_{{\sf SGDI}}$ accepts with probability $\geq \frac{1}{c^{1/4}}$.
In this case, by Theorem \ref{thm:state-transfer}, 
the composite system ${\sf R}_A\otimes{\sf R}_B$ of the outputs of ${\cal P}_{{\sf SGDI}}$ 
has a quantum state $\tilde{\rho}$ 
such that 
\[
(\langle\psi_A^f|\langle\psi_B^f|)\tilde{\rho}(|\psi_A^f\rangle|\psi_B^f\rangle)\geq 
1-\frac{1}{c^{1/4}},
\]
under the condition that ${\cal P}_{{\sf SGDI}}$ accepts. 
By Lemma~\ref{lem:FvG},
\begin{align*}
D(|\psi_A^f\rangle\langle\psi_A^f|\otimes|\psi_B^f\rangle\langle\psi_B^f|, \tilde{\rho})
&\leq 
\sqrt{1-F(|\psi_A^f\rangle\langle\psi_A^f|\otimes|\psi_B^f\rangle\langle\psi_B^f|, \tilde{\rho})^2}\\
&= \sqrt{1- (\langle\psi_A^f|\langle\psi_B^f|)\tilde{\rho}(|\psi_A^f\rangle|\psi_B^f\rangle)
}\\
&\leq \frac{ 1 }{c^{1/8}}.
\end{align*}
Therefore, 
by the analysis of the ideal case, 
the SWAP test of $v_r$ accepts with probability at most 
\[
\frac{1}{2} + 2 \left(\frac{\ell (r+1)}{|\mathbb{F}|}\right)^2+\frac{ 1 }{c^{1/8}}
\leq \frac{1}{2} + 2 \left(\frac{1}{\tilde{c}\cdot 2^{\log|U|}}\right)^2+\frac{ 1 }{c^{1/8}},
\]
which is at most $\frac{1}{2}+\gamma$ for some small constant $\gamma>0$ when we choose sufficiently large $c$ and $\tilde{c}$. 
Now the acceptance probability $v_{{\sf acc}}$ of ${\cal P}_{{\sf seteq}}$ is 
\[
v_{{\sf acc}}
=
\Pr[\mbox{the SWAP test of $v_r$ on input $\tilde{\rho}$ accepts}]
\times
\Pr[{\cal P}_{{\sf SGDI}}\ \mbox{accepts}],
\]
which is also at most $\frac{1}{2}+\gamma$ for some small constant $\gamma>0$ when we choose sufficiently large $c$ and $\tilde{c}$. 
Then, we can reduce soundness to any small constant 
by parallel AND-type repetitions, namely, by running ${\cal P}_{{\sf seteq}}$ a suitable constant number of times in parallel and accepting when all the runs accept. (This soundness reduction follows from the standard parallel repetition of QMA. See, e.g.,~\cite{kitaev2002classical}.)
This completes the proof of Theorem~\ref{thm:SetEQ}.
\end{proof}

\subsection{Classical bounds for $\SE$}\label{subsec:classical_bound_for_se}
Here we prove the classical lower bounds shown in Theorem~\ref{thm:SetEQ_classical_lower_bound}.\\

\noindent\textbf{Proof of Theorem~\ref{thm:SetEQ_classical_lower_bound}}:
In order to prove our lower bounds for $\SE$,
we utilize reductions from $\mathsf{EQ}_n^2$ to $\SE$, then apply the following lower bound of $\mathsf{EQ}_n^2$ that appears in~\cite{FLNP21}.
\begin{lem}[Theorem 9 of~\cite{FLNP21}]\label{lemma:lower_bound_for_eq}
Let $r\geq 3$ be a positive integer. Consider an instance of $\mathsf{EQ}_n^2$ where the two nodes $v_0$ and $v_r$ on a line graph $v_0,\ldots,v_r$ are provided with inputs $x\in \{0,1\}^n$ and $y \in \{0,1\}^n$. Then, for any dMA protocol that solves $\mathsf{EQ}_n^2$ for this instance with completeness $1-p$ and soundness $1-2p-\varepsilon$ for any $p,\varepsilon>0$, there exists $i\in \{1,\ldots, r-1\}$ such that the certificate size of $v_i$ is $\Omega(n)$.
\end{lem}
Let $\mathcal{P}$ be a dMA protocol for $\SE$ with the certificate size $s_c$ which appears in the statement of the theorem. We show three different reductions depending on the size of $|U|$.\\ 

\noindent \textbf{Case 1 $|U|<\ell$:}
We consider the following instance of $\mathsf{EQ}_n^2$ on a line graph of length $r\geq 3$ for $n=(|U|-1)\cdot \log (\ell/|U|)$: the node $v_0$ is provided with $x$, and the node $v_r$ is provided with $y$. 
The input $x$ corresponds to the input $(a_{v_0,1},\ldots,a_{v_0,\ell})$ of $\SE$ in the following way: The input $x$ is split into $|U|-1$ substrings $x_1,\ldots, x_{|U|-1}$ of length $\lceil \log (\ell/|U|) \rceil$. Let $U= \{u_1,\ldots , u_{|U|-1},u_{|U|}\}$ be the universal set. Let $c(x_i)$ be the number whose binary representation is $x_i$ for each $i \in \{1,\ldots,|U|-1\}$. We interpret the number of $u_i$'s $v_0$ holds as $c(x_i)$ and the number of $u_{|U|}$'s $v_0$ holds as $\ell - \sum_i c(x_i)$ so that the total number of elements $v_0$ holds is $\ell$. (Note that the total number of elements held by $v_0$ does not exceed $\ell$ since $\sum_i c(x_i)\leq (|U|-1)\cdot \ell/|U| \leq \ell$.) The input $(b_{v_0,1},\ldots,b_{v_0,\ell})$ is $(u_{|U|},\ldots,u_{|U|})$. The input for $v_r$ is constructed in the same way, but in this case the input $(a_{v_r,1},\ldots,a_{v_r,\ell})$ is $(u_{|U|},\ldots,u_{|U|})$ and the input $(b_{v_r,1},\ldots,b_{v_r,\ell})$ is determined by $y$. For internal nodes $v_1,\ldots, v_{r-1}$, all input elements are $u_{|U|}$ (i.e., each internal node is provided with $2\ell$ identical elements). Now the numbers of $u_i$'s $v_0$ and $v_r$ hold are the same if and only if $x_i = y_i$ for each $i \in \{1,\ldots,|U|-1\}$. Now we can use the protocol $\mathcal{P}$ to solve $\mathsf{EQ}_n^2$ for $n = (|U|-1)\cdot \log (\ell/|U|)$. Thus by Lemma~\ref{lemma:lower_bound_for_eq}, we have $s_c = \Omega(|U|\log(\ell/|U|))$.\\

\noindent \textbf{Case 2 $|U|=\Omega(\ell)$:}
We consider the following instance of $\mathsf{EQ}_n^2$ on a line graph $v_0,\ldots,v_r$ of length $r\geq 3$: the node $v_0$ is provided with $x$, and the node $v_r$ is provided with $y$. 
Let $\ell$ be the minimum integer satisfying
$
\binom{3\ell}{\ell} \geq 2^n.
$
Since $\binom{3\ell}{\ell} = 2^{3\ell H(1/3) - O(\log \ell)}$ where $H(\cdot)$ is the binary entropy function, we have $\ell = \Theta(n)$. Let $S=\{s\in\{0,1\}^{3\ell}: |s|=\ell\}$ be the set of $3\ell$-bit strings so that $|S|=\binom{3\ell}{\ell}$. We arbitrarily choose one injection $f:\{0,1\}^n \rightarrow S$ (this kind of injection exists since we have $|S|\geq 2^n$).
The network constructs the following instance of $\mathsf{SetEquality}_{\ell,U}$ for the universal set $U=\{0,1,2,\ldots, 3\ell\}$ without communication:
\begin{itemize}
\item The inputs $x$ and $y$ are converted to $f(x),f(y)\in S$. Let $X=\{i:f(x)_i = 1\}$ and $Y=\{i:f(y)_i = 1\}$ be two sets of $\ell$ elements from the universal set $\{1,2,\ldots,3\ell\}$. $X$ and $Y$ are regarded as the inputs $(a_{v_0,1},\ldots,a_{v_0,\ell})$ and $(b_{v_r,1},\ldots,b_{v_r,\ell})$ of $\mathsf{SetEquality}_{\ell,U}$. Furthermore, we set $(b_{v_0,1},\ldots,b_{v_0,\ell}) = (a_{v_r,1},\ldots,a_{v_r,\ell}) = (0,\ldots, 0)$. 
\item The inputs to each internal node $v_1,\ldots,v_{r-1}$ are set to $(0,\ldots, 0),(0,\ldots, 0)$. 
\end{itemize}
Now the set $A$ and $B$ of this instance of $\mathsf{SetEquality}_{\ell,U}$ are identical as multisets if and only if $f(x)=f(y)$. Since $f$ is an injection, we have $f(x)=f(y) \Leftrightarrow x=y$ and thus the output of $\mathsf{SetEquality}_{\ell,U}$ on this instance is identical to that of $\mathsf{EQ}_n^2$ on the input $x$ and $y$. Now we can use the protocol $\mathcal{P}$ to solve $\mathsf{EQ}_n^2$ for $n = \Theta(\ell)$. Thus by Lemma~\ref{lemma:lower_bound_for_eq}, we have $s_c = \Omega(\ell)$.\\

\noindent \textbf{Case 3 $|U| = \Omega(r\ell)$:}
%If $r = O(1)$, then Lemma~\ref{lemma:eq_to_se_2} is essentially equivalent to Lemma~\ref{lemma:eq_to_se}. We thus consider the case that $r = \omega(1)$.
Without loss of generality, assume that $r$ is odd so that $r = 2k + 1$ for some positive integer $k$. Let $\{v_0,\ldots , v_{2k+1}\}$ be the set of nodes. Consider the following situation: for any $i\in\{0,1,\ldots,k-1\}$, the node $v_i$ receives the input $x_i\in\{0,1\}^n$ and the node $v_{2k+1-i}$ receives the input $y_i\in\{0,1\}^n$. There is no input for $v_k,v_{k+1}$. For these inputs, the goal is to decide iff $x_i=y_i$ for all $i\in\{0,1,\ldots,k-1\}$. The answer is thus identical to $\prod_{i}\mathsf{EQ}_n^2(x_i,y_i)$. Now we construct the following instance of $\mathsf{SetEquality}_{\ell,U}$ for $|U|=3k\ell + 1$: 
\begin{itemize}
\item According to the proof for Case 2, we convert the input $(x_i,y_i)$ to the inputs $(a_{v_i,1},\ldots,a_{v_i,\ell}),(b_{v_i,1},\ldots,b_{v_i,\ell})$ and $(a_{v_{2k+1-i},1},\ldots,a_{v_{2k+1-i},\ell}),(b_{v_{2k+1-i},1},\ldots,b_{v_{2k+1-i},\ell})$ of $\mathsf{SetEquality}_{\ell,U_i}$ for $\ell = \Theta(n)$ where the universal set is $U_i=\{0,3i\ell + 1,\ldots,3i\ell + 3\ell\}$;
\item The input to the node $v_i$ for $i=k,k+1$ is $ (a_{v_i,1},\ldots, a_{v_i,\ell}) = (b_{v_i,1},\ldots, b_{v_i,\ell}) = (0,\ldots,0)$.
\end{itemize}
The goal is to decide if $A=B$ as multisets for 
\begin{align*}
&A=\bigl\{a_{v_i,j}:i\in\{0,\ldots,2k+1\}, j\in \{1,\ldots,\ell\}\bigr\},\\
&B=\bigl\{b_{v_i,j}:i\in\{0,\ldots,2k+1\}, j\in \{1,\ldots,\ell\}\bigr\}.
\end{align*}
The output is thus the same as that of $\prod_{i}\mathsf{EQ}_n^2(x_i,y_i)$. This is because the inputs $(a_{v_i,1},\ldots,a_{v_i,\ell})$ and $(b_{v_i,1},\ldots,b_{v_i,\ell})$ of $v_i$, which are the multisets from the universe $\{3i\ell + 1,\ldots,3i\ell + 3\ell\}$ unless they are $(0,\ldots,0)$, are disjoint from the inputs of $v_{i'}$ for any $v_{i'}\neq v_{i}$.

We consider the instance of $\mathsf{EQ}_{kn}^2(x_0x_1\cdots x_{k-1},y_0y_1\cdots y_{k-1})$ on the four node line graph $\{u_0,u_1,u_2,u_3\}$ where $u_0$ possesses the input $x_0x_1\cdots x_{k-1}$ and $u_3$ possesses the input $y_0y_1\cdots y_{k-1}$. The network can use the protocol $\mathcal{P}$ to solve this $\mathsf{EQ}_{kn}^2$ instance where $|U| = 3k\ell + 1 = \Theta(r\ell)$. Here the node $u_0$ simulates $k$ nodes $\{v_0,\ldots, v_{k-1}\}$, $u_1$ simulates $v_k$, $u_2$ simulates $v_{k+1}$, and $u_3$ simulates $k$ nodes $\{v_{k+2},\ldots, v_{2k+1}\}$. The certificate size of $u_i$ is then $ks_c$ for $i=0,3$, otherwise $s_c$. Now it is followed by Lemma~\ref{lemma:lower_bound_for_eq} that $s_c=\Omega(kn) = \Omega(r\ell)$.
 \qed \\

Finally we prove the classical upper bound for $\SE$.\\

\noindent\textbf{Proof of Theorem~\ref{thm:SetEQ_classical_upper_bound}}:
The upper bound of $O(r\ell \log |U|)$ is easily achieved by the trivial protocol, i.e., the prover sends $A$ and $B$ (which are represented by $r\ell \lceil \log |U|\rceil $ bits) to each node, then each node (1) checks if its input is consistent with the corresponding part of the certificate, (2) checks if the whole certificate is the same as those of its adjacent nodes, and (3) performs local computation to check if $A=B$. Now we show the another upper bound. Let $U=\{u_1,\ldots,u_{|U|}\}$ be the universal set and $\{v_0,\ldots, v_r\}$ be the set of nodes. Since the number of elements in the multisets $A$ and $B$ are both $r\ell$, the number of $u_i$ in $A$ and $B$ for each $i\in \{1,\ldots,|U|\}$ can both be represented by $\lceil \log(r\ell) \rceil$-bit. Consider the following protocol: The honest prover sends each node $v_i$ the number of each $u_j$'s appearing in the input $\{a_{v_k,k'}:k\in \{0,\ldots, i\}, k'\in \{1,\ldots,\ell\}\}$ and the number of $u_j$'s appearing in the input $\{b_{v_k,k'}:k\in \{0,\ldots, i\}, k'\in \{1,\ldots,\ell\}\}$ using $O(|U|\log(r\ell))$ bits. The node $v_0$ can locally check the consistency of its certificate. Assume that the certificate of $v_{i-1}$ is consistent with the honest prover's certificate. The node $v_i$ can check the consistency of its certificate from its input and the certificate of $v_i$. Therefore, if the prover sends any illegal certificate at least one node can recognize it. If the prover sends legal certificate (i.e., the certificate assignment that is accepted by all $v_i$ for $i\in \{0,\ldots, r-1\}$), the node $v_r$ can construct $A$ and $B$ locally and thus can decide the correct output.\qed

%%%%%%%%%%%%%%%%%%%%%%%%%%%%%%%%%%%%%%%
\section{Conversion of dQMA protocols into LOCC dQMA protocols}\label{subsec:ZH}
%%%%%%%%%%%%%%%%%%%%%%%%%%%%%%%%%%%%%%%
%\subsection{Distributed verification of EPR pairs by LOCC}\label{subsec:ZH}

In this section we show how to create an EPR pair $|\Phi^+\rangle=\frac{1}{\sqrt{2}}(|00\rangle+|11\rangle)$ between two parties without quantum communication in the setting where a prover helps the nodes in a non-interactive way. Our protocol is based on the verification protocol of the EPR pair in the adversarial setting proposed by Zhu and Hayashi~\cite{ZH19PRA} 
(see also \cite{ZH19PRL,ZH19PRA-2}), 
who showed that a verifier $V$ can check whether a two-qubit state sent from a (possibly malicious) prover is $|\Phi^+\rangle$. 

The following is the verification protocol given in Ref.~\cite{ZH19PRA}. 
%(the protocol itself is essentially the same as the protocol described in Ref.~\cite{PLM18PRL} but Ref.~\cite{ZH19PRA} showed that it can be used even in the adversarial setting).

\fboxsep=6pt
\begin{breakbox}
\noindent
{\bf Protocol ${\cal P}_{{\sf ZH}}$}: 
Let %$\sigma_1,\sigma_2,\ldots,\sigma_N,\sigma_{N+1}$ be $(N+1)$ 
${\sf R}_1,{\sf R}_2,\ldots,{\sf R}_N, {\sf R}_{N+1}$ be $(N+1)$ two-qubit registers from the prover. 
Here, $|+\rangle:=\frac{1}{\sqrt{2}}(|0\rangle+|1\rangle)$, $|-\rangle:=\frac{1}{\sqrt{2}}(|0\rangle-|1\rangle)$, $|+'\rangle:=\frac{1}{\sqrt{2}}(|0\rangle+i|1\rangle)$ and $|-'\rangle:=\frac{1}{\sqrt{2}}(|0\rangle- i|1\rangle)$. 

\begin{enumerate}
\item Perform a random permutation $\pi$ on the $(N+1)$ two-qubit registers, 
and rename ${\sf R}_j:={\sf R}_{\pi(j)}$ for $j=1,2,\ldots,N+1$.
\item For each $j=1,2,\ldots,N$, the verifier $V$ does one of the following three POVMs on register ${\sf R}_{j}$ with probability $1/3$ for each: 
\begin{itemize}
    \item $M_1=\{E_1,I-E_1\}$ with $E_1=|00\rangle\langle 00|+|11\rangle\langle 11|$. 
    \item $M_2=\{E_2,I-E_2\}$ with $E_2=|++\rangle\langle ++|+|--\rangle\langle --|$. 
    \item $M_3=\{E_3,I-E_3\}$ with $E_3=|+'-'\rangle\langle +'-'|+|-'+'\rangle\langle -'+'|$.
\end{itemize}
\item Reject if the second components in the POVMs are obtained. Otherwise, the test passes and outputs ${\sf R}_{N+1}$.
\end{enumerate}
\end{breakbox}%\vspace{3mm}
\fboxsep=3pt

Ref.~\cite{ZH19PRA} describes $E_1=\frac{I+Z^{\otimes 2}}{2}$, $E_2=\frac{I+X^{\otimes 2}}{2}$ and $E_3=\frac{I-Y^{\otimes 2}}{2}$, 
where 
\[X=\left(\begin{matrix}
0 & 1\\
1 & 0
\end{matrix}
\right), Y=\left(\begin{matrix}
0 & -i\\
i & 0
\end{matrix}
\right),\text{ and }Z=
\left(\begin{matrix}
1 & 0\\
0 & -1
\end{matrix}
\right),
\]
while we rewrite them as above (which is a similar expression to the protocol in Ref.~\cite{PLM18PRL}) 
since it would be easy to see for our purpose.
Importantly, 
step 2 implements the POVM $\{\Omega, I-\Omega\}$ on ${\sf R}_{j}$ with $\Omega=\frac{2}{3}|\Phi^+\rangle\langle\Phi^+|+\frac{1}{3} I$ for each $j$, but it is implemented by local measurements on each qubit of ${\sf R}_{j}$. 

The following result was shown for the protocol ${\cal P}_{{\sf ZH}}$ in Ref.~\cite{ZH19PRA}.

\begin{theorem}[Zhu-Hayashi]\label{thm:ZhuHayashi}
There is a number 
$N=O(\frac{1}{\varepsilon}\log(\frac{1}{\delta}))$ such that if the test passed with probability 
at least $\delta$, 
then the output state $\tilde{\sigma}$ of ${\cal P}_{{\sf ZH}}$ (under the condition that the test passes)
satisfies $\langle \Phi^+|\tilde{\sigma}|\Phi^+\rangle\geq 1-\varepsilon$.
% see 062335-5 (ZH19PRA 100 062335) B.Main figures of merit for what delta is
%To verify the target state with fidelity $1-\varepsilon$ and significance level $\delta$, $N=O(\frac{1}{\varepsilon}\log(\frac{1}{\delta}))$ is sufficient.
\end{theorem}
    
The protocol ${\cal P}_{{\sf ZH}}$ uses only local measurements, and thus, 
it can be used for verifying the sharing of an EPR pair by two parties 
who only use local operations and classical communication 
(LOCC). 

Let $V_1$ and $V_2$ be neighboring parties who expect to receive $|\Phi^+\rangle$ jointly from the prover.   
The following protocol is a simple implementation 
of ${\cal P}_{{\sf ZH}}$ with LOCC by $V_1$ and $V_2$. \vspace{2mm}

%The protocol ${\cal P}_{{\sf ZH}}$ is very suitable to use distributed verification between $V_1$ and $V_2$ in the LOCC style as follows.

\fboxsep=6pt
\begin{breakbox}
\noindent
{\bf Protocol ${\cal P}_{{\sf ZHLOCC}}$}: 
Let %$\sigma_1,\sigma_2,\ldots,\sigma_N,\sigma_{N+1}$ be $(N+1)$ 
${\sf R}_{1,1}, \ldots, {\sf R}_{N,1}, {\sf R}_{N+1,1}$ be $(N+1)$ one-qubit registers from the prover to $V_1$, 
and ${\sf R}_{1,2}, \ldots, {\sf R}_{N,2}, {\sf R}_{N+1,2}$ be $(N+1)$ one-qubit registers from the prover to $V_2$, respectively. 
%Here, $|+'\rangle:=\frac{1}{\sqrt{2}}(|0\rangle+i|1\rangle)$ and $|-'\rangle:=\frac{1}{\sqrt{2}}(|0\rangle- i|1\rangle)$. 

\begin{enumerate}
\item $V_1$ chooses a random permutation $\pi$ on $\{1,2,\ldots,N+1\}$ and sends it to $V_2$, and then both perform $\pi$ 
on the $(N+1)$ two-qubit registers $({\sf R}_{1,1}, {\sf R}_{1,2}),\ldots, ({\sf R}_{N,1}, {\sf R}_{N,2}), ({\sf R}_{N+1,1}, {\sf R}_{N+1,2})$. Rename ${\sf R}_{j,1}:={\sf R}_{\pi(j),1}$ and ${\sf R}_{j,2}:={\sf R}_{\pi(j),2}$. 
\item $V_1$ chooses $N$ random numbers $k_1,k_2,\ldots,k_N \in \{1,2,3\}$ and sends them to $V_2$.  
For each $j=1,2,\ldots,N$, $V_1$ and $V_2$ implement one of the POVMs $M_1,M_2,M_3$ on register $({\sf R}_{j,1}, {\sf R}_{j,2})$ jointly as follows.
\begin{itemize} 
\item when $k_j=1$, they jointly implement $M_1=\{E_1,I-E_1\}$; $V_1$ and $V_2$ measure ${\sf R}_{j,1}$ and  ${\sf R}_{j,2}$ in the $Z$ basis $\{|0\rangle,|1\rangle\}$, respectively, 
and $V_1$ sends the measurement value to $V_2$, who rejects iff it differs from the measurement value of $V_2$. 
\item when $k_j=2$, they jointly implement $M_2=\{E_2,I-E_2\}$; $V_1$ and $V_2$ measure ${\sf R}_{j,1}$ and  ${\sf R}_{j,2}$ in the $X$ basis $\{|+\rangle,|-\rangle\}$, respectively, 
and $V_1$ sends the measurement value to $V_2$, who rejects iff it differs from the measurement value of $V_2$. 
\item when $k_j=3$, they jointly implement $M_3=\{E_3,I-E_3\}$; $V_1$ and $V_2$ measure ${\sf R}_{j,1}$ and  ${\sf R}_{j,2}$ in the $Y$ basis $\{|+'\rangle,|-'\rangle\}$, respectively, and $V_1$ sends the measurement value to $V_2$, who rejects iff it is same as the measurement value of $V_2$.  
%(ii) $M_2=\{E_2,I-E_2\}$ with $E_2=|++\rangle\langle ++|+|--\rangle\langle --|$; and 
%(iii) $M_3=\{E_3,I-E_3\}$ with $E_3=|+'-'\rangle\langle +'-'|+|-'+'\rangle\langle -'+'|$.
\end{itemize}
%If $V_2$ rejects, he/she reports it to $V_1$.
\item The test passes and $V_1$ and $V_2$ output ${\sf R}_{N+1,1}$ and ${\sf R}_{N+1,2}$, respectively. 
\end{enumerate}
\end{breakbox} 
\fboxsep=3pt
\vspace{3mm}
It is easy to see that ${\cal P}_{{\sf ZHLOCC}}$ simulates ${\cal P}_{{\sf ZH}}$ exactly in a distributed manner. 
The protocol ${\cal P}_{{\sf ZHLOCC}}$ does not use 
any quantum communication between $V_1$ and $V_2$, 
while the amount of classical communication used 
between $V_1$ and $V_2$ is 
$
\lceil \log (N+1)! \rceil + \lceil \log 3^N\rceil + N = O(N\log N).
$

Furthermore, we can replace a random permutation $\pi$ in step 1 
of ${\cal P}_{{\sf ZHLOCC}}$ 
by switching the $j$th two-qubit register $({\sf R}_{j,1},{\sf R}_{j,2})$
and the $(N+1)$th register $({\sf R}_{N+1,1},{\sf R}_{N+1,2})$ 
by choosing $j$ uniformly at random from $\{1,2,\ldots,N+1\}$ 
(actually, doing nothing when $j=N+1$) 
since the output state by such change is the same as protocol ${\cal P}_{{\sf ZHLOCC}}$.  
We call the protocol by such change ${\cal P}^+_{{\sf ZHLOCC}}$.
Now the amount of classical communication used 
between $V_1$ and $V_2$ in ${\cal P}^+_{{\sf ZHLOCC}}$ is improved to 
$
\lceil \log (N+1) \rceil + \lceil \log 3^N\rceil + N = O(N).
$

Thus the following theorem holds for ${\cal P}^+_{{\sf ZHLOCC}}$.

\begin{theorem}\label{thm:distributedZH}
For the same number $N=O(\frac{1}{\varepsilon}\log(\frac{1}{\delta}))$ as Theorem~\ref{thm:ZhuHayashi},  
if the test passed with at least probability $\delta$, 
then the two-qubit state $\tilde{\sigma}$ output by $V_1$ and $V_2$ in ${\cal P}^+_{{\sf ZHLOCC}}$ 
satisfies $\langle \Phi^+|\tilde{\sigma}|\Phi^+\rangle\geq 1-\varepsilon$. 
\end{theorem}  
 
%Now by using protocol ${\cal P}^+_{{\sf ZHLOCC}}$ with the number $N$ of Theorem~\ref{thm:distributedZH}, $V_1$ and $V_2$ can teleport one qubit without any shared EPR pair but with a certificate of size $s_{c}^{{\sf ZH}} = O(\frac{1}{\varepsilon}\log(\frac{1}{\delta}) )$ from the prover and with a classical message of size $s_{m}^{{\sf ZH}} = O(\frac{1}{\varepsilon}\log(\frac{1}{\delta}) )$. 
%In Appendix~\ref{appendix:analysis}, we prove Theorem~\ref{thm:convertion-LOCC} by showing how to use the protocol ${\cal P}^+_{{\sf ZHLOCC}}$.

Now we prove Theorem~\ref{thm:convertion-LOCC} by showing how to use the above protocol in order to convert any dQMA protocol into an LOCC dQMA protocol.
\begin{proof}[Proof of Theorem~\ref{thm:convertion-LOCC}]
Let ${\cal P}$ be the original dQMA protocol on $G=(V,E)$. Let $C_{uv}$ be the number of qubits sent from $u\in V$ to $v\in N(u)$ in ${\cal P}$.
Let $\mathsf{M}_u$ be the quantum register 
from the prover to $u$ in ${\cal P}$. 
Our approach is to create $C_{uv}$ EPR-pairs $\ket{\Phi^+}^{\otimes C_{uv}}$ (more precisely, a quantum state that is sufficiently close to the desired EPR-pairs) between $u$ and $v$ using ${\cal P}^+_{{\sf ZHLOCC}}$. According to Theorem~\ref{thm:distributedZH}, we get $\tilde{\sigma}$ satisfying $\bra{\Phi^+}\tilde{\sigma}\ket{\Phi^+} \leq 1- \varepsilon$ using $O(\frac{1}{\varepsilon}\log(\frac{1}{\delta}))$ qubits. Since we need $s^{\mathcal{P}}_{tm}$ EPR-pairs, the amount of errors is at most $s^{\mathcal{P}}_{tm}\varepsilon$. Taking $\varepsilon = O(1/s^{\mathcal{P}}_{tm})$, we get a quantum state sufficiently close to $\ket{\Phi^+}^{\otimes C_{uv}}$. The total number of qubits needed as the certificate for quantum teleportation between each edge $(u,v)$ is $C_{uv}\cdot s^{\mathcal{P}}_{tm}$.
We formalize this approach as a protocol ${\cal P}_{{\sf LOCC}}$ as follows, and analyze completeness, soundness, the certificate size and the message size of ${\cal P}_{{\sf LOCC}}$.
 \vspace{3mm}
\fboxsep=6pt
\begin{breakbox}
\noindent
{\bf Protocol ${\cal P}_{{\sf LOCC}}$}: 
\begin{enumerate}
\item For each $u\in V$, the following is done:
\begin{enumerate}
\item Node $u$ receives $\mathsf{M}_u$ and $(N+1)C_{uv}$-qubit registers ${\sf R}_{u,v}$, which consist of  $C_{uv}$ blocks of $(N+1)$-qubit registers 
for each $v\in N(u)$ from the prover.
\item For each $v\in N(u)$, $u$ and $v$ do as follows:
\begin{itemize}
\item To create $j$th EPR pair $(1\leq j\leq C_{uv})$ between $u$ and $v$, implement protocol ${\cal P}^+_{{\sf ZHLOCC}}$ on the $j$th block of ${\sf R}_{u,v}$. 
Let ${\sf Q}_{u,v,j}$ be the two-qubit register obtained by ${\cal P}^+_{{\sf ZHLOCC}}$, where one is owned by $u$ and the other is owned by $v$ 
(if ${\cal P}^+_{{\sf ZHLOCC}}$ passes).  
\end{itemize}
\end{enumerate}
\item Simulate the verification stage of ${\cal P}$ 
using ${\sf M}:=\otimes_{u\in V} {\sf M}_u$, 
where sending the $j$th qubit from $u$ to $v$ is implemented 
by the quantum teleportation protocol~\cite{BBC+93PRL} using ${\sf Q}_{u,v,j}$.
\end{enumerate}
\end{breakbox}\vspace{3mm}
\fboxsep=3pt

\noindent
{\bf Analysis:} 
Completeness holds trivially 
since the prover can send the certificates of ${\cal P}$ as the contents of ${\sf M}$, and the EPR pairs for quantum teleportation 
to keep the completeness parameter $p_c$ of ${\cal P}$. 
Thus we consider soundness.

For the soundness $p_s$ of the original protocol ${\cal P}$, 
we set $\delta=p_s+\gamma$ for each run of ${\cal P}^+_{{\sf ZHLOCC}}$ 
implemented in protocol ${\cal P}_{{\sf LOCC}}$, where $\gamma$ is the parameter appearing in the statement of Theorem~\ref{thm:convertion-LOCC}.

Now we can assume that 
each run of ${\cal P}^+_{{\sf ZHLOCC}}$ 
passes with probability at least $\delta$ 
under the condition that the other runs passed 
since otherwise ${\cal P}_{{\sf LOCC}}$ accepts 
with probability $< p_s+\gamma$, 
which satisfies the required soundness $p_s+\gamma$ 
of ${\cal P}_{{\sf LOCC}}$. 

Let $\tilde{\sigma}$ be the total state of all the registers sent from the prover after passing the tests of ${\cal P}^+_{{\sf ZHLOCC}}$, 
and $\tilde{\sigma}^{\mathrm{EPRs}}$ 
be the reduced state of $\tilde{\sigma}$ 
on the registers 
\[
{\sf Q}:=\bigotimes_{u\in V} \bigotimes_{v\in N(u)} 
\bigotimes_{j\in\{1,2,\ldots,C_{uv}\}} {\sf Q}_{u,v,j}.
\] 
Then, 
by Theorem \ref{thm:distributedZH} and the above assumption, 
the reduced state of $\tilde{\sigma}$ 
on ${\sf Q}_{u,v,j}$, 
$\tilde{\sigma}_{u,v,j}$ 
satisfies
\begin{equation}\label{eq:C-analysis}
\tr(P^- \tilde{\sigma}_{u,v,j})\leq \varepsilon,
\end{equation}   
where $P^-:= I-|\Phi^+\rangle\langle\Phi^+|$. 
By Lemma~\ref{lem:union-bound} and the inequality~(\ref{eq:C-analysis}), 
we obtain
\begin{align}
\lefteqn{\tr\left(
\bigotimes_{u\in V} \bigotimes_{v\in N(u)} \bigotimes_{j\in\{1,2,\ldots,C_{uv}\}} (|\Phi^+\rangle\langle\Phi^+|)_{{\sf Q}_{u,v,j}}\tilde{\sigma}^{\mathrm{EPRs}}
\right)}\quad\nonumber\\
&\geq 1 - \sum_{u\in V,v\in N(u)}\sum_{j\in \{1,2,\ldots,C_{uv} \}} \tr(
(I_{{\sf Q}\setminus{\sf Q}_{u,v,j}}\otimes (P^-)_{{\sf Q}_{u,v,j}})\tilde{\sigma}^{\mathrm{EPRs}} )
)\nonumber\\
&= 
1 - \sum_{u\in V,v\in N(u)}\sum_{j\in \{1,2,\ldots,C_{uv} \}} \tr( (P^-)_{{\sf Q}_{u,v,j}}\tilde{\sigma}_{u,v,j} ))\nonumber\\
&\geq 1- s_{tm}^{{\cal P}} \varepsilon, \label{eq:union}
\end{align}
where the identity $I_{
{\sf Q}\setminus{\sf Q}_{u,v,j}}$ 
acts on the registers of ${\sf Q}$ except ${\sf Q}_{u,v,j}$. 

The acceptance probability $p_{\mathsf{acc}}$ of ${\cal P}_{{\sf LOCC}}$ is
\[
p_{\mathsf{acc}} = \tr(M_x \tilde{\tau})\times\Pr[\mbox{all runs of ${\cal P}^+_{{\sf ZHLOCC}}$ pass}],
\]
where $M_x$ is the POVM element that corresponds to the acceptance by the simulation of the verification stage of ${\cal P}$ (step 2 in ${\cal P}_{{\sf LOCC}}$), 
and 
$\tilde{\tau}$ 
is the state on the system ${\sf Q}\otimes {\sf M}$ 
after passing all runs of ${\cal P}^+_{{\sf ZHLOCC}}$. 

Let $|\mathrm{EPRs}\rangle=\bigotimes_u\bigotimes_{v\in N(u)}\bigotimes_{j=1}^{C_{uv}} |\Phi^+\rangle_{{\sf Q}_{u,v,j}}$.  
By Lemma~\ref{lem:MHNF15}, the inequality~(\ref{eq:union}), and Lemma~\ref{lem:FvG}, 
there is some state $\xi$ on ${\sf M}$ 
such that
\begin{align*}
D(\tilde{\tau},|\mathrm{EPRs}\rangle\langle \mathrm{EPRs}|\otimes\xi)
&\leq \sqrt{1-F(|\mathrm{EPRs}\rangle\langle \mathrm{EPRs}|\otimes\xi,\tilde{\tau})^2}\\
&=\sqrt{1-F(|\mathrm{EPRs}\rangle\langle \mathrm{EPRs}|,\tr_{{\sf M}}(\tilde{\tau} ))^2}\\
&\leq \sqrt{s_{tm}^{{\cal P}}\varepsilon}.
\end{align*}
Therefore, 
\begin{align*}
p_{\mathsf{acc}} 
&\leq \tr(M_x\tilde{\tau})\\
&\leq \tr(M_x (|\mathrm{EPRs}\rangle\langle \mathrm{EPRs}|\otimes\xi))+\sqrt{s_{tm}^{{\cal P}}\varepsilon}\\
&\leq p_s + \sqrt{s_{tm}^{{\cal P}}\varepsilon}.
\end{align*}
Now we take $\varepsilon=\frac{\gamma^2}{s_{tm}^{{\cal P}}}$, 
which implies that protocol ${\cal P}_{{\sf LOCC}}$ accepts with probability at most $p_s+\gamma$, as required.

%Let $d_{\max}$ be the maximum degree of the network, 
%and $C$ be the maximum of the qubits communicated between any two nodes in ${\cal P}$.  
Finally, we consider the certificate size and the message size 
of ${\cal P}_{{\sf LOCC}}$. 
The qubits sent from the prover to each node $u$ 
are the qubits sent in ${\cal P}$ 
and the qubits sent in ${\cal P}^+_{{\sf ZHLOCC}}$ to create $\sum_{v\in N(u)} (C_{uv}+C_{vu})$ EPR pairs.
Thus, the certificate size of ${\cal P}_{{\sf LOCC}}$ is at most
\[
s_c^{{\cal P}} + d_{\max} s_m^{{\cal P}} s_{c}^{{\sf ZH}} = s_c^{{\cal P}} + O(d_{\max} s_m^{{\cal P}} s_{tm}^{{\cal P}}).
\]
The bits communicated between any two neighboring nodes $u$ and $v$ 
are the two bits sent in quantum teleportation and the bits sent in ${\cal P}^+_{{\sf ZHLOCC}}$ 
for each qubit communicated in ${\cal P}$. 
Thus, the message size of  ${\cal P}_{{\sf LOCC}}$ is   
\[
s_m^{{\cal P}} (2+s_{m}^{{\sf ZH}}) = O(s_m^{{\cal P}} s_{tm}^{{\cal P}}).
\]
%Now the proof of Theorem~\ref{thm:convertion-LOCC} is completed.
\end{proof}

\bibliography{dQMA-bib}

\end{document}